\documentclass[12pt]{amsart}

\usepackage{amsmath,amssymb,amsthm}

\usepackage{amsmath,amstext,amssymb,amsopn,amsthm}
\usepackage{hyperref}
\usepackage{mathtools}
\usepackage{enumerate,bbm}
\usepackage[section]{placeins}
\usepackage{color,graphicx}
\usepackage{url,cases}

\usepackage[margin=30mm]{geometry}

\usepackage{graphicx,tikz, caption}

\allowdisplaybreaks

\newtheorem{theorem}{Theorem}[section]

\newtheorem{lemma}[theorem]{Lemma}

\theoremstyle{definition}

\newtheorem{step}[theorem]{Step}

\theoremstyle{remark}
\newtheorem{remark}[theorem]{Remark}

\newcommand{\bone}{\mathbbm{1}}

\newcommand{\E}{\operatorname{\mathbb{E}}}
\renewcommand{\P}{\operatorname{\mathbb{P}}}
\newcommand{\prt}{\partial}
\newcommand{\wh}{\widehat}
\newcommand{\R}{\mathbb{R}}

\newcommand{\calF}{\mathcal{F}}
\newcommand{\calE}{\mathcal{E}}
\newcommand{\bal}{{\boldsymbol\alpha}}
\newcommand{\bbet}{{\boldsymbol\beta}}

\newcommand{\wt}{\widetilde}
\newcommand{\eps}{\varepsilon}

\numberwithin{equation}{section}

\begin{document}

\title[Pinned billiard balls]{From pinned billiard balls to\\ partial differential equations}

\author[Burdzy]{Krzysztof Burdzy}
\email{burdzy@uw.edu}
\address{Department of Mathematics, University of Washington, Seattle,
WA 98195}

\author[Hoskins]{Jeremy G. Hoskins}
\email{ jeremyhoskins@uchicago.edu}
\address{Department of Statistics, University of Chicago,
IL 60637}

\author[Steinerberger]{Stefan Steinerberger}
\email{steinerb@uw.edu}
\address{Department of Mathematics, University of Washington, Seattle,
WA 98195}

\thanks{KB’s research was supported in part by Simons Foundation Grants 506732 and 928958. SS's research was partially supported by the NSF (DMS-2123224) and the
Alfred P. Sloan Foundation.}



\begin{abstract} We discuss the propagation of kinetic energy through billiard balls fixed in place along a one-dimensional segment.
The number of billiard balls is assumed to be large but finite and we assume kinetic energy propagates following the usual collision laws of physics. Assuming an underlying stochastic mean-field for the expectation and the variance of the kinetic energy, we derive a coupled system of nonlinear partial difference equations. Our results are illustrated by numerical simulations.
\end{abstract}

\maketitle

\section{Introduction}

This paper is concerned with the evolution of pseudo-velocities of ``pinned billiard balls'' introduced in \cite{fold}. Pinned billiard balls do no move but they have pseudo-velocities which evolve according to the usual totally elastic collision laws for velocities of moving balls. 
We will take a step towards an ``approximate'' hydrodynamic limit model and the corresponding nonlinear partial difference equations.
In Section \ref{a9.1} we  describe the pinned billiard balls model in detail, we  present a conjecture stating its large scale behavior (modulated white noise hypothesis), we  derive partial difference equations for the parameters of modulated white noise, and we  indicate how partial difference equations lead to nonlinear partial differential equations.
Section \ref{a9.2} is devoted to numerical results supporting the modulated white noise hypothesis.
Section \ref{a9.5} contains the discussion of the basic properties of the PDEs informally derived in Remark \ref{a17.1} and ends with the discussion of some hydrodynamic limit results in the literature.
Section \ref{a17.2}  contains the proof of the main rigorous mathematical result of this paper on partial difference equations.

\section{Evolution of pinned billiard balls model parameters}\label{a9.1}

We will present some computations inspired by a one-dimensional system of pinned billiard balls, a special case of a model introduced in \cite{fold}. 
In a system of pinned billiard balls, the  balls touch some other balls and have pseudo-velocities but they do not move. The balls ``collide,'' i.e., their pseudo-velocities change according to the usual laws of totally elastic collisions.

In our case, the centers of the balls are arranged on a finite segment of the real line. Their centers are one unit apart and their radii are all equal to $1/2$, so there is a finite ordered set of balls, each touching its two neighbors (except for the two endpoints, where the balls have only one neighbor). 
See Fig. \ref{fig1}.

\begin{center}
\begin{figure}[h!]
\begin{tikzpicture}
\draw [thick] (0,0) circle (0.5cm);
\draw [thick] (1,0) circle (0.5cm);
\draw [thick] (2,0) circle (0.5cm);
\draw [thick] (3,0) circle (0.5cm);
\draw [thick] (4,0) circle (0.5cm);
\draw [thick] (5,0) circle (0.5cm);
\end{tikzpicture}
\caption{Billard balls arranged along a one-dimensional line. The balls touch but are fixed for all time.}\label{fig1}
\end{figure}
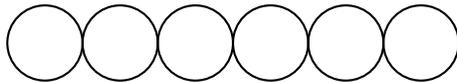
\end{center}

The spacetime for the model is discrete, i.e.,
the velocities $v(x,t)$ are defined for $x=1,2,\dots, n$ and $t=0,1,2, \dots$, where $x$ is the position (i.e., number) of the $x$-th ball.
The evolution, i.e., pseudo-collisions of the balls and transformations of the velocities, is driven by an exogenous random process because the balls do not move and hence they cannot collide in the usual way.

First consider a simplified model in which pairs of adjacent balls are chosen randomly, i.e., in a uniform way, and form an i.i.d. sequence. Every time a pair of adjacent balls is chosen, the velocities become ordered, i.e., if the chosen balls have labels $x$ and $x+1$ and the collision occurs at time $t$ then
\begin{align*}
v(x,t+1) &= \min (v(x,t), v(x+1, t)), \\
v(x+1,t+1) &= \max (v(x,t), v(x+1, t)).
\end{align*}
This agrees with the usual transformation rule for velocities of moving balls of equal masses undergoing totally elastic collisions.
The evolution described above  has been studied under the names of
``random sorting networks'' in \cite{AHRV},
``oriented swap process'' in \cite{AHR} and ``TASEP speed process'' in \cite{AAV}. It has been also called ``colored TASEP.''
For a related model featuring confined (but moving) balls, see \cite{Gaspard_2,Gaspard_2008_1,Gaspard_2008}.\\

While the model described above is very natural and well motivated by physics, it is characterized by a property that is strictly limited to the one-dimensional collision systems---the set of all initial velocities is conserved. The velocities are only rearranged.  In multidimensional pinned ball families energy packets will not be preserved. A typical collision will change two energy packets into two new energy packets of different sizes subject to obeying the conservation laws. \\

The model described below is a compromise between the one-dimensional and higher dimensional models. It is one-dimensional to make the analysis easier but it involves energy exchange to simulate multidimensional evolutions.
Our model incorporates an idea from \cite{BBO2006}, \cite[p. 70]{BBO2009} or \cite[Sect. 2.1.2]{JKO}.
In the following model the evolution of velocities in the one-dimensional family of pinned balls consists of a sequence of two-step transformations. 
In the first step we redistribute energy. In the second step we reorder a pair of velocities.
We start by generating an i.i.d. sequence $(x_t,\, t=0,1,2,\dots)$, with each $x_t$ distributed uniformly in $\{2,3,\dots, n-1\}$. 

\begin{step}\label{a29.1}
Suppose that velocities $v(x,s)$ have been defined for $s=0,\dots,t$ and all $x=1,2,\dots , n$. Consider the following equations for $v_-(x_t-1,t+1), v_-(x_t,t+1)$ and $v_-(x_t+1,t+1)$, representing conservation of energy and momentum,
\begin{align}\label{a28.3}
&v_-(x_t-1,t+1)+ v_-(x_t,t+1) + v_-(x_t+1,t+1)\\
 &\qquad=
v(x_t-1,t)+ v(x_t,t) + v(x_t+1,t),\notag\\
&v_-(x_t-1,t+1)^2+ v_-(x_t,t+1)^2 + v_-(x_t+1,t+1)^2\label{a28.4} \\
&\qquad =
v(x_t-1,t)^2+ v(x_t,t)^2 + v(x_t+1,t)^2.\notag
\end{align}
Given $v(x_t-1,t), v(x_t,t)$ and $v(x_t+1,t)$, the set of solutions $(v_-(x_t-1,t+1), v_-(x_t,t+1),v_-(x_t+1,t+1))$ forms a circle in three-dimensional space, since it is the intersection of a sphere with a two-dimensional plane. We use extra randomness, independent of everything else, to choose a point  $(v_-(x_t-1,t+1), v_-(x_t,t+1),v_-(x_t+1,t+1))$ uniformly on this circle. This completes the first step.
\end{step}

\begin{step}\label{a29.2}
In the second step,
the above energy exchange is followed by reordering of a pair of velocities.
Let $\kappa_t $ be equal $-1$ or $1$, with equal probabilities, independent of everything else. If $\kappa_t = -1$  then
\begin{align*}
v(x_t-1,t+1) &= \min (v_-(x_t-1,t+1), v_-(x_t, t+1)), \\
v(x_t,t+1) &= \max (v_-(x_t-1,t+1), v_-(x_t, t+1)),\\
v(x_t+1, t+1) & = v_-(x_t+1,t+1).
\end{align*}
Otherwise,
\begin{align*}
v(x_t-1, t+1) & = v_-(x_t-1,t+1),\\
v(x_t,t+1) &= \min (v_-(x_t,t+1), v_-(x_t+1, t+1)), \\
v(x_t+1,t+1) &= \max (v_-(x_t,t+1), v_-(x_t+1, t+1)).
\end{align*}
This completes the second step.
\end{step}

For all $x\ne x_t-1, x_t, x_t+1$, we let $v(x,t+1) = v(x,t)$.

\subsection{Modulated white noise}\label{a9.3}

We will assume
that the joint distribution of $\{v(x,t), 1\leq x\leq n, t\geq 0\}$ converges after appropriate rescaling to
\begin{align}\label{a28.1}
v(x,t) = \mu(x,t) + \sigma(x,t) W(x,t),
\end{align}
when $n$ goes to infinity. Here
$W(x,t)$ is spacetime white noise
and $\mu(x,t)$ and $\sigma(x,t)$ are deterministic functions.

We will present numerical evidence for our assumption in Section \ref{a9.2}.
On the theoretical side, our assumption is questionable 
(see Section \ref{rev1.7}) but we will defend it in Section \ref{rev1.8}
as a reasonable compromise between true hydrodynamic model and simplicity.

In our discrete model, ``white noise'' is a collection of i.i.d. standard normal random variables.
The way we will formally work with this assumption is to note that
$$ \mathbb{E}~ v(x,t) = \mu(x,t)$$
and
$$ \mathbb{E}~ v(x,t)^2 = \mu(x,t)^2 + \sigma(x,t)^2$$
from which it becomes possible to deduce both the values of $\mu$ and $\sigma$.

On the theoretical side, there is an immense literature on interacting particle systems and hydrodynamic limits. We list only one article and two books: \cite{Bert,KL,Spohn}. The model closest to ours seems to be that of ``hot rods'' introduced in \cite{DF,BDS}, but the balls are allowed to move in that model.
For recent results on the hot rods model and a review of the related literature, see \cite{FerEt,FerOl}.
The linear version of our partial differential equations  for $\mu$ and $\sigma$ is essentially the same as the equations (2.1) and (2.2)  for the local density  and the local current in \cite{Bert}.
We will discuss some results from \cite{TV03} in Section \ref{rev1.7}.\\

\subsection{Partial difference equations}\label{a12.1}

Proving the hydrodynamic limit theorem for the pinned balls model is a major technical challenge. In this paper, we limit ourselves to a very modest step. We will derive  formulas for the one-step evolution of parameters $\mu$ and $\sigma$ under the assumption that \eqref{a28.1} holds.

We will use the following notation. 
For functions $\wt \mu: (0,1)\times[0,\infty) \to \R$ and $\wt \sigma: (0,1) \times [0,\infty)\to (0,\infty)$ we let
for $t\geq 0$ and $1\leq x \leq n$,
\begin{align*}
&\mu(x,t) = \wt \mu\left( \frac{x-1}{n-1},  \frac{t}{n}\right)\qquad\qquad
\sigma(x,t) = \wt \sigma\left( \frac{x-1}{n-1},  \frac{t}{n}\right),\\
&\wt\mu_{x,n}(x,t)= \frac {\prt} {\prt z} \wt\mu(z,s)\Big|_{\substack{z=(x-1)/(n-1)\\s=t/n}}\qquad
\qquad
\wt\sigma_{x,n}(x,t)= \frac {\prt} {\prt z} \wt\sigma(z,s)\Big|_{\substack{z=(x-1)/(n-1)\\s=t/n}}\\
&\wt\mu_{xx,n}(x,t)= \frac {\prt} {\prt z^2} \wt\mu(z,s)\Big|_{\substack{z=(x-1)/(n-1)\\s=t/n}}\\
&\left(\wt\mu(x,t)^2
+ \wt\sigma(x,t)^2\right)_{xx,n}= \frac {\prt} {\prt z^2} \left(\wt\mu(z,s)^2
+ \wt\sigma(z,s)^2\right)\Big|_{\substack{z=(x-1)/(n-1)\\s=t/n}}\\
&\wt\mu_{t,n}(x,t)=
\frac {\prt} {\prt s} \wt\mu(z,s)\Big|_{\substack{z=(x-1)/(n-1)\\s=t/n}}\ .
\end{align*}

\begin{theorem}\label{a28.2}
Suppose functions $\wt \mu: (0,1)\times[0,\infty) \to \R$ and $\wt \sigma: (0,1) \times [0,\infty)\to (0,\infty)$ are $C^3_b$ (with bounded third derivative). 
Assume that  $v(x,t)$ satisfies \eqref{a28.1} at a fixed time $t\geq 0$ and for $1\leq x \leq n$. 
Then for $2\leq x \leq n-1$,
 \begin{align}\label{m12.10}
\left(n-2\right)
&\E (v(x, t+1)-v(x,t))\\
&=  -
\frac 1 {\sqrt{ \pi}(n-1)}
  \wt\sigma_{x,n}(x,t)
+ \frac 2 {(n-1)^2}  \wt\mu_{xx,n}(x,t)
 + O(n^{-3}), \notag\\
\left(n-2\right)
&\E (v(x, t+1)^2-v(x,t)^2)\label{m12.11}\\
 &=  -
\frac 2 {\sqrt{ \pi}(n-1)}
\left( \mu(x,t)
  \wt\sigma_{x,n}(x,t)
+ \sigma(x,t)
 \wt\mu_{x,n}(x,t)
\right)\notag\\
&\quad+ \frac 2 {(n-1)^2}
\left(\wt\mu(x,t)^2
+ \wt\sigma(x,t)^2\right)_{xx,n}
 + O(n^{-3}) .\notag
 \end{align}
\end{theorem}

\begin{remark} \label{a17.1}
We will outline how \eqref{m12.10}-\eqref{m12.11} lead to a system of PDEs.

We have not proved that 
\eqref{a28.1} is the limiting distribution of the system when $n\to\infty$ but we expect that
\begin{align*}
\E& (v(x, t+1)-v(x,t)) = \mu(x, t+1) - \mu(x,t) \\
&= \wt\mu(x/n, t/n+1/n) - \wt\mu(x/n,t/n)
\approx
\wt\mu_{t,n}(x,t).
\end{align*}
We combine this with \eqref{m12.10} to obtain
 \begin{align*}
&\left(n-2\right)
\wt\mu_{t,n}(x,t)
\approx  -
\frac 1 {\sqrt{ \pi}(n-1)}
 \wt\sigma_{x,n}(x,t)
+ \frac 2 {(n-1)^2}  \wt\mu_{xx,n}(x,t)
 + O(n^{-3}) .\notag
 \end{align*}
We  rescale space and time, multiply both sides by $(n-1)^2/2$ and ignore the error terms so that
 \begin{align*}
&\frac{(n-1)^2(n-2)}2
\frac {\prt} {\prt t} \wt\mu(x,t)=  -
\frac {n-1} {2\sqrt{ \pi}}
 \frac {\prt} {\prt x} \wt\sigma(x,t)
+  \frac {\prt^2} {\prt x^2} \wt\mu(x,t) .\notag
 \end{align*}
We remove the factor $\frac{(n-1)^2(n-2)} 2 $ from the left hand side  by rescaling time and we obtain 
 \begin{align*}
&
\frac {\prt} {\prt t} \wt\mu(x,t)=  -
\frac {n-1} {2\sqrt{ \pi}}
 \frac {\prt} {\prt x} \wt\sigma(x,t)
+  \frac {\prt^2} {\prt x^2} \wt\mu(x,t) .\notag
 \end{align*}

A completely analogous argument starting with \eqref{m12.11} yields
\begin{align*}
 \frac \prt{\prt t} &
\left(\wt\mu(x,t)^2 + \wt\sigma(x,t)^2\right)\\
&=   \frac {\prt^2} {\prt  x^2} \left(\wt\mu(x,t)^2 + \wt\sigma(x,t)^2\right)
-   \frac {n-1 }{\sqrt{ \pi}}
\left( \wt\mu(x,t) \frac {\prt} {\prt  x} \wt\sigma(x,t)
+ \wt\sigma(x,t) \frac {\prt} {\prt  x} \wt\mu(x,t)\right).\notag
\end{align*}

We set  
\begin{align}\label{a13.1}
\lambda = \frac {n-1} {2\sqrt{ \pi}}
\end{align}
and change the notation from $\wt\mu$ and $\wt\sigma$ to $\mu$ and $\sigma$ to obtain the following  form of these equations, 
\begin{align}\label{m29.5}
\mu_t &= \Delta_x \mu
- \lambda \sigma_x,\\
(\sigma^2 + \mu^2)_t&=  \Delta_x (\sigma^2 + \mu^2)
- 2\lambda (\mu \sigma)_x.\label{m29.6}
\end{align}

For boundary conditions, we take
\begin{align}
\sigma(a) = \sigma(b) &= 0.\label{m29.8}
\end{align}
We offer a heuristic justification for \eqref{m29.8}. Billiard ball velocities become ordered and stay ordered at the endpoints of the system because there are no constraints on one side preventing the increasing ordering of the velocities at the endpoints of the configuration. Hence, $\sigma$ instantaneously becomes and stays equal to zero at the endpoints. It is not obvious that \eqref{m29.8} are sufficient for uniqueness of solutions, but we know that uniqueness holds for a related (simplified) set of equations in \cite{KBJS}.

\end{remark}

\begin{remark}\label{rev1.6}
A straightforward formal calculation transforms the system \eqref{m29.5}-\eqref{m29.6}into
\begin{align*}
 \begin{cases}
\mu_t &=  \Delta \mu_{} - \lambda \sigma_x,\\
\sigma_t &= \Delta \sigma_{} -  \lambda \mu_x + \frac 1 \sigma (\sigma_x^2  +\mu_x^2 ),
\end{cases}
\end{align*}
and these equations can be rescaled as
\begin{align*}
 \begin{cases}
\frac{1}{\lambda}\mu_t &=  \frac{1}{\lambda}\Delta \mu_{} - \sigma_x,\\
\frac{1}{\lambda}\sigma_t &= \frac{1}{\lambda}\Delta \sigma_{} -  \mu_x + \frac{1}{\lambda}\cdot\frac 1 \sigma (\sigma_x^2  +\mu_x^2 ).
\end{cases}
\end{align*}

The parameter $\lambda$ should be thought of as large. It represents the spatial size of the discrete system, so in simulations it typically takes a value larger than $1,000$.
Hence, the following limiting case is of interest,
\begin{align}\label{rev1.3}
 \begin{cases}
\frac{1}{\lambda}\mu_t &=  - \sigma_x,\\
\frac{1}{\lambda}\sigma_t &=  -  \mu_x.
\end{cases}
\end{align}
After rescaling time by $\lambda$ we obtain,
\begin{align}\label{rev1.2}
 \begin{cases}
\mu_t &=  - \sigma_x,\\
\sigma_t &=  -  \mu_x.
\end{cases}
\end{align}
This implies 
\begin{align}\label{m29.2}
     \mu_{tt} \sim  \mu_{xx} \qquad \mbox{and} \qquad  \sigma_{tt} \sim  \sigma_{xx}
\end{align}
 suggesting that for large values of $\lambda$, the dynamics might be similar to that of the wave equation whenever the nonlinearity is small. 
 
However, we also note that $\sigma \geq 0$ which limits the extent to which the wave equation analogy can be applied. The paper \cite{KBJS} is devoted exclusively to studying the equations \eqref{rev1.2} with the constraint $\sigma \geq 0$.

The paper \cite{KBAO} contains a result on the existence of solutions to \eqref{rev1.2} with given terminal values. 
\end{remark}

\section{Numerical Examples}\label{a9.2}

We will show the results of 100,000 simulations with 1,000 balls and compare them to numerical
solutions of PDEs \eqref{m29.5}-\eqref{m29.6} obtained using a standard finite difference method with 3201 equispaced spatial grid points.
First we show figures supporting the conjecture that the pinned balls system
is represented by modulated white noise. Then we will present numerical
evidence for the agreement between the evolution of parameters $\mu$ and $\sigma$ in the
collision model and the PDEs.

We will present the results for only one representative set of initial conditions, namely,
\begin{align*}
\mu(x,0) &= {\rm Erf}\left(27(x-1/2)^3\right),\\
\sigma(x,0)&= \frac{[1-\cos(2\pi x)]}{1000},
\end{align*}
for $x\in[0,1],$ where Erf denotes the error function.
The functions were rescaled from the interval $[0,1]$ to $[1, 1000]$
for the collision simulations.

We define $T$ as the time  until the
apparent total freeze, i.e.,  the time when the variance $\sigma^2$ is almost
identically equal to 0 (compared to typical values in the main part of the evolution).
We have approximately $T\approx  0.000120162$ for the timescale used in  
\eqref{m29.5}-\eqref{m29.6} with $\lambda = (n-1) / (2 \sqrt{\pi})$ (to match \eqref{a13.1})
and $n=1,000$ (the number of balls). 

The distribution of the empirical white noise $W$ at time $0.37T$ estimated from the simulations  matches
the normal distribution quite well, according to Fig. \ref{figa12.2}. The values of the empirical white noise were
calculated for a spatial position $x$ by subtracting the mean $\mu(x,0.37T)$ and dividing by the standard deviation $\sigma(x,0.37T)$,
where the last two functions were
evaluated as averages over all runs.
\begin{figure} [h]
\includegraphics[width=0.5\linewidth]{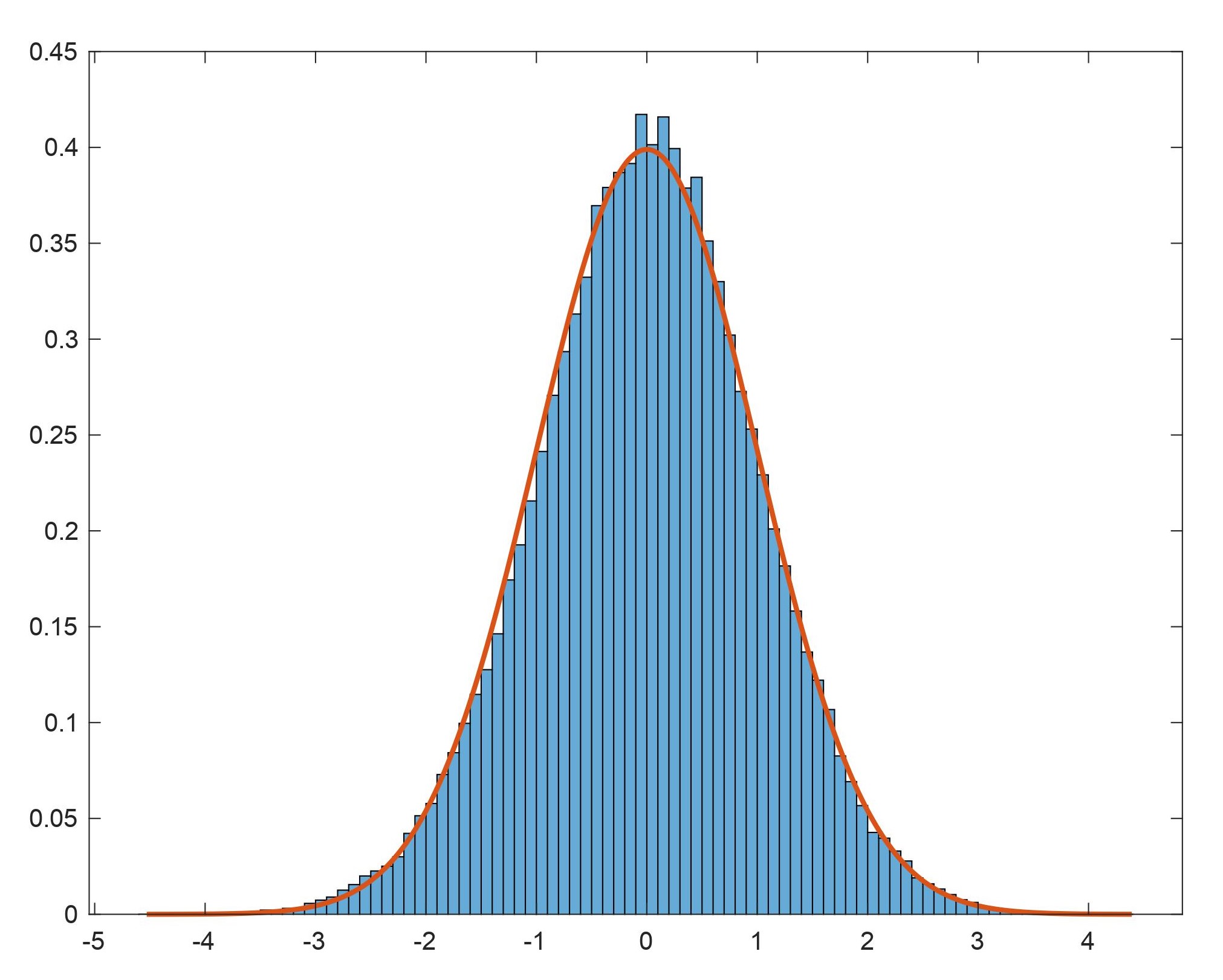}
\caption{
Empirical histogram (blue) of white noise values at the time $0.37T$.  The standard normal
density is drawn in red.
}
\label{figa12.2}
\end{figure}

Correlations of the adjacent velocities should be equal to 0 assuming the white noise hypothesis.
If we have a sample of size $n$ from the bivariate standard normal distribution then the density of the empirical correlation coefficient is
\begin{align*}
f(r) = \frac{(1-r^2)^{(n-4)/2}}{B(1/2,(n-2)/2},
\end{align*}
where $B$ is the beta function. The standard deviation of this distribution is $1/\sqrt{n-1}$. For $n=1,000$, the standard deviation is about $0.032$.
 We show in Fig. \ref{figa12.4} that correlation values are not much larger than the theoretical value.
\begin{figure} \includegraphics[width=0.7\linewidth]{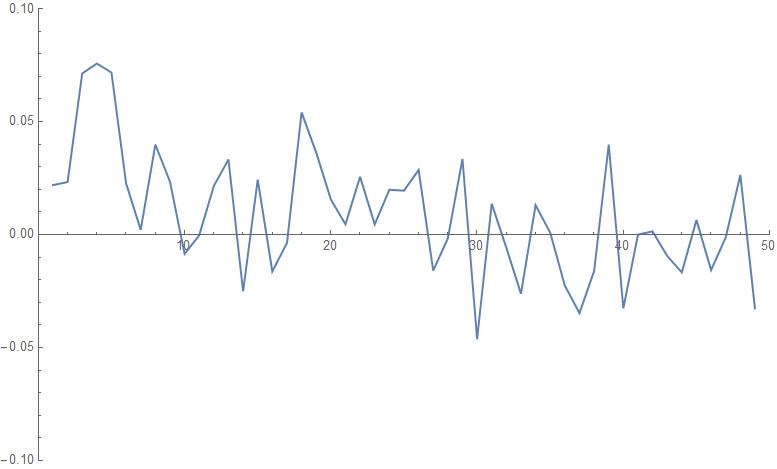}
\caption{
 Correlation between white noise values at the distance $k$ at time $0.37 T$, for $k=1,\dots , 50$,  for a single run. The values of the empirical white noise were
calculated for a spatial position $x$ by subtracting the mean $\mu(x,0.37T)$ and dividing by the standard deviation $\sigma(x,0.37T)$,
where the last two functions were
evaluated as averages over all runs.
}
\label{figa12.4}
\end{figure}

Fig. \ref{figa12.6} shows that the joint distribution of the noise at adjacent sites is rotationally
symmetric, as expected from white noise. The color is added to improve perception. 
\begin{figure} \includegraphics[width=0.5\linewidth]{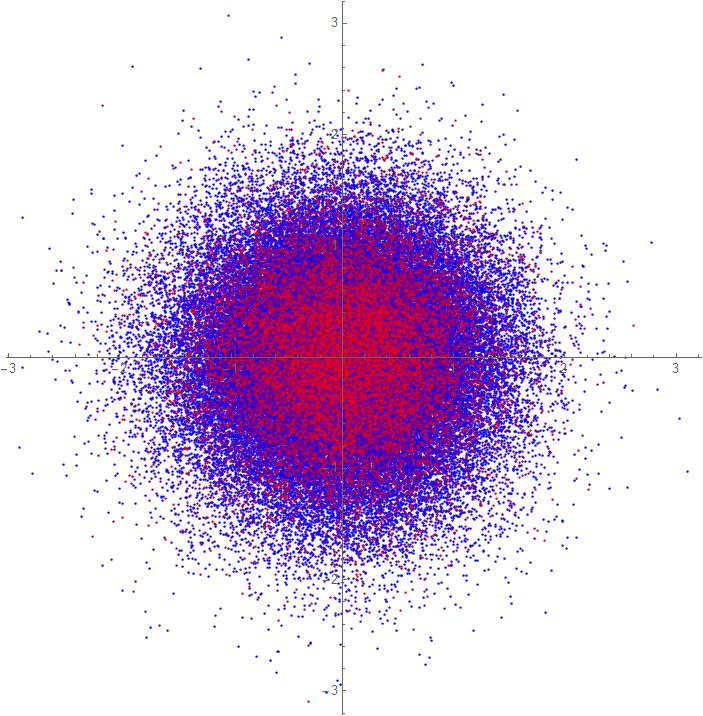}
\caption{
 Pairs of values of the noise $(W(300,0.37T),W(301,0.37T))$ for 100,000 runs of the simulation. The  RGB scheme is
  $((k/100,000)^5,0, 1-(k/100,000)^5)$ where $k$ is the number of the simulation.
}
\label{figa12.6}
\end{figure}

Fig. \ref{figa12.1} shows different stages of the evolution of $\mu$ and $\sigma$.
The agreement between the moments $\mu$ and $\sigma$ estimated from simulations and the solutions to PDEs \eqref{m29.5}-\eqref{m29.6}
is excellent. 
Fig. \ref{figa12.1}  supports our choice of the initial conditions---the resulting evolution of $\mu$ and $\sigma$
has interesting complexity.
\begin{figure}[h]
    \centering
    \begin{minipage}{0.33\textwidth}
        \centering
        \includegraphics[width=0.99\textwidth]{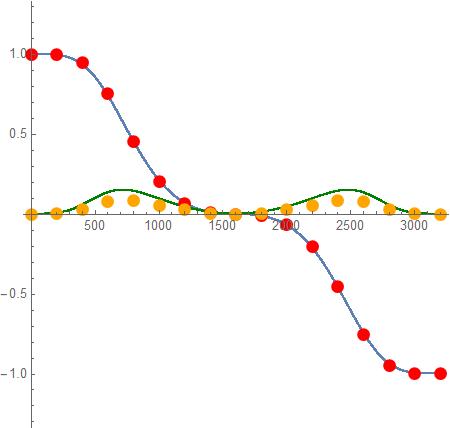}
    \end{minipage}\hfill
    \begin{minipage}{0.33\textwidth}
        \centering
        \includegraphics[width=0.99\textwidth]{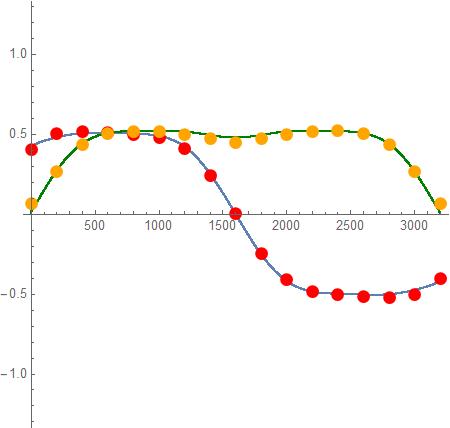}  
    \end{minipage}
    \begin{minipage}{0.33\textwidth}
        \centering
        \includegraphics[width=0.99\textwidth]{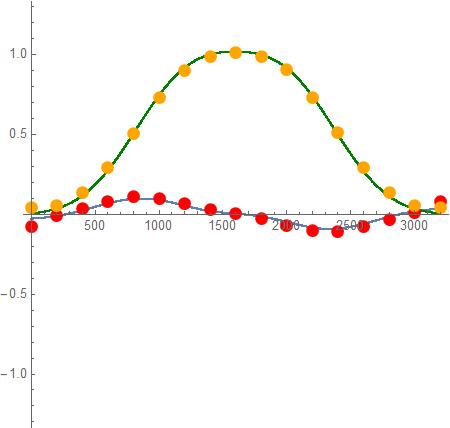}  
    \end{minipage}
    \\
\begin{minipage}{0.33\textwidth}
        \centering
        \includegraphics[width=0.99\textwidth]{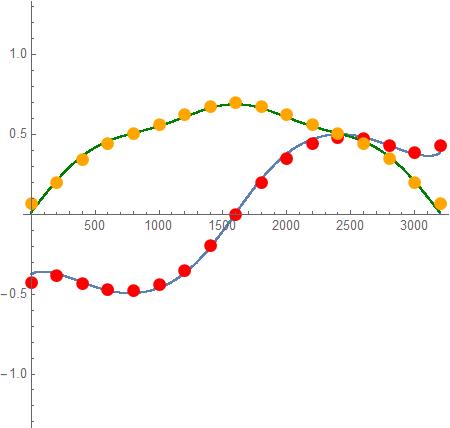}  
    \end{minipage}\hfill
    \begin{minipage}{0.33\textwidth}
        \centering
        \includegraphics[width=0.99\textwidth]{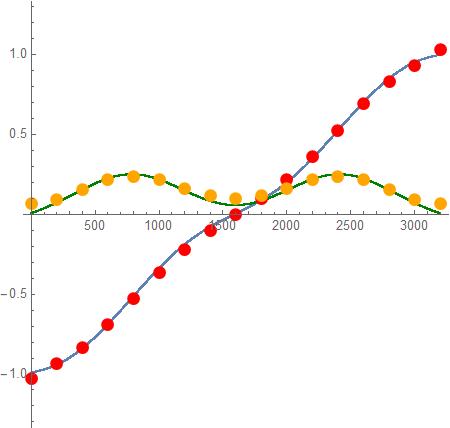}  
    \end{minipage}
    \begin{minipage}{0.33\textwidth}
        \centering
        \includegraphics[width=0.99\textwidth]{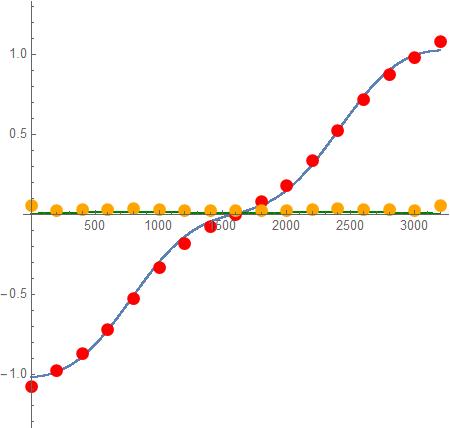}  
    \end{minipage}
    \caption{The moments $\mu$ and $\sigma$ at times $0.01 T,0.2 T,0.4 T,0.6 T,0.8 T$ and $0.99 T$ (top from left to right, then bottom left to right). 
 The mean $\mu$ (dotted red) and $\sigma$ (dotted orange) were estimated by averaging values over 100,000 repetitions of the pinned balls model.
 The mean $\mu$ (solid blue) and $\sigma$ (solid green) were numerically computed using the equations \eqref{m29.5}-\eqref{m29.6}.
 The curves were horizontally and vertically rescaled to show agreement.} \label{figa12.1}
\end{figure}

\section{A System of PDEs}\label{a9.5}

  \subsection{Different coordinates }\label{a17.4}
 There is  another representation 
of our PDEs that we will need for the discussion later in this section.
Let
 $$ E(x,t) = \frac{1}{2}\left(\mu(x,t)^2 +\sigma(x,t)^2\right).$$
  The equation for $\mu$ can then be written as
  $$ \mu_t - \Delta \mu = - \lambda (\sqrt{2E - \mu^2})_x.$$
A computation shows
\begin{align*}
 E_t &= \mu \mu_t + \sigma \sigma_t \\
 &= \mu (\Delta \mu - \lambda \sigma_x) + \sigma \Delta \sigma - \lambda \sigma \mu_x + (\sigma_x)^2 + (\mu_x)^2\\
 &= \mu \Delta \mu + (\mu_x)^2 + \sigma \Delta \sigma + (\sigma_x)^2 - \lambda (\mu \cdot \sigma )_x \\
 &= \Delta E - \lambda (\mu \sqrt{2E - \mu^2})_x.
 \end{align*}
 This leads to the following version of \eqref{m29.5}-\eqref{m29.6},
 \begin{align}\label{rev1.1}
\begin{cases}
 \mu_t - \Delta \mu &=   - \lambda (\sqrt{2E - \mu^2})_x ,\\
  E_t -\Delta E &= - \lambda (\mu \sqrt{2E - \mu^2})_x ,  \\
  E(a,t) &= \mu(t,a)^2/2,\\
  E(b,t) &= \mu(t,b)^2/2. 
  \end{cases}
  \end{align}

\subsection{Related hydrodynamic models}\label{rev1.7}

We are grateful to Balint T\'oth for the following remarks on related hydrodynamic limits  (but we take responsibility for possible inaccuracies).

As shown in \cite{TV03} there are some combinatorial conditions for the rates of local dynamics
in interacting particle systems with several conservation laws, in order that the ergodic
translation invariant measures be of product structure (see condition (C) in \cite{TV03}). That
paper is formulated in the context of discrete local observables but all arguments apply to continuous observables. In the
present context this  implies that in order to get the product Gaussians as Gibbs
measures the rates of swaps $(v(x,t), v(x+1,t)) \to (v(x+1,t), v(x,t))$ must be chosen as
$(r(v(x,t))-r( v(x+1,t)))\bone_{\{v(x,t)\geq v(x+1,t)\}}$
where $r: \R\to \R$ is a  non-decreasing and non-constant function (see section 2 of \cite{TV03}). A natural  choice could be $r(u)=u$. 
This would lead to 
the PDE's
 \begin{align}\label{rev1.5}
 \begin{cases}
 \mu_t - \Delta \mu &=   - \lambda (2E - \mu^2)_x ,\\
  E_t -\Delta E &= - \lambda (\mu 2E - \mu^2)_x ,
  \end{cases}
  \end{align}
rather than our \eqref{rev1.1}:
 \begin{align*}
\begin{cases}
 \mu_t - \Delta \mu &=   - \lambda (\sqrt{2E - \mu^2})_x ,\\
  E_t -\Delta E &= - \lambda (\mu \sqrt{2E - \mu^2})_x . 
  \end{cases}
  \end{align*}
  
The Onsager relations discussed in \cite{TV03} say that the thermodynamic Gibbs entropy (of the equilibrium measures of the microscopic system) expressed as a function of the conserved quantities is a Lax entropy for the PDE.  This  feature links the microscopic system (Gibbs entropy of the stationary/ergodic measures) with the macroscopic hydrodynamic PDE (Lax entropy). This feature fails to hold in our system and ansatz. 
It has been shown later in \cite{GS}that the Onsager relations are valid in wider generality than in \cite{TV03}. Namely, there is no need for the product structure of the stationary measure.  However, in the more general setting it is more difficult (if possible at all) to find explicit formulas for the Gibbs entropy. 

\subsection{Justification of our model}\label{rev1.8}
Although the results reviewed in Section \ref{rev1.7} indicate that our model cannot lead to ``modulated white noise'' local equilibrium measure, we will try to justify our approach.

(i) If we repeat Step \ref{a29.1} multiple times, say $n^\alpha$ times, for every single Step \ref{a29.2}, this will basically affect only $\lambda$, and it will change its value from order $n$ to $n^{1-\alpha}$. Assuming $\alpha\in(0,1)$, we can expect that the local equilibrium measure will be close to Gaussian because Step \ref{a29.1} is locally a mixing process on the sphere.

The postulate that white noise governs the evolution of $v(x,t)$ is motivated by the postulate of equidistribution of energy. Due to conservation of energy,  $\sum_{x\in N} v(x,t)^2$ is more or less constant over small time intervals in a small neighborhood $N$. Hence, for a fixed $t$, one expects the vector $\{v(x,t), x\in N\}$ to be approximately uniformly distributed over the sphere and, therefore, to be approximately i.i.d. normal. 

(ii) Step \ref{a29.2} may introduce correlation between values of $v$ at various locations for a fixed time. Specifically, we may expect negative correlation between adjacent sites because Step \ref{a29.2} orders the values in the increasing manner. However, Step \ref{a29.2} is associated with the energy redistribution in Step \ref{a29.1} on three adjacent sites. Hence, one can expect that the correlation between the first two and the last two sites in the triplet would annihilate the correlation effects in the difference equation calculations. One can expect much lower correlation effects from the sites at least two units away.

Simulations discussed in Section \ref{a9.2} (see especially Fig. \ref{figa12.4}) suggest that correlations are small. 

(iii)
Although our assumption that the local equilibrium measure is modulated white noise is questionable in view of the results  reviewed in Section \ref{rev1.7}, the excellent agreement between simulations of $\mu$ and $\sigma$ and numerical solutions of our PDEs shown in Fig. \ref{figa12.1} justifies our model as a very good, if not perfect, approximation of the true hydrodynamic limit.

(iv)
Starting with \eqref{rev1.5} and
proceeding as in Remark \ref{rev1.6}, in particular,
 ``sending $\lambda\to\infty$'' as in \eqref{rev1.3}-\eqref{rev1.2}, we would obtain a system of Burgers-like equations
\begin{align}\label{rev1.4}
 \begin{cases}
\mu_t &=  - \sigma_x \sigma,\\
\sigma_t &=  -  \mu_x \sigma.
\end{cases}
\end{align}
rather than the transport equations \eqref{rev1.2}. The equations \eqref{rev1.4} develop shocks (see \cite{Smoller,Serre}) and, therefore, present a major technical challenge for the hydrodynamic limit theory.
On the other hand, the transport equations \eqref{rev1.2} combined with the freezing condition, i.e., $\mu_t=\mu_x=\sigma_t=\sigma_x=0$ whenever $\sigma=0$, are tractable and display interesting behavior---this has been shown in \cite{KBJS}. Hence, the present approach seems to have intrinsic pure mathematical value independent of the physical applications.

(v) The paper \cite{FT} is about a very closely related system and the corresponding hydrodynamic limit, going even beyond the shocks. That model can be represented in our terms as follows: particles have three types of velocities, $-1, 0, +1$. Otherwise the dynamics is very similar to ours: momentum and kinetic energy are conserved, particles move according to their velocities, with only ``monotone'' swaps allowed. But the rate of near-neighbor swaps depends on the values of the velocities, unlike in our system. The system of PDEs obtained in the hydrodynamic limit is the so-called Leroux-system. These PDEs develop shocks, unlike our limiting PDEs. 

\section{Proof of partial difference equations}\label{a17.2}

\begin{proof}[Proof of Theorem \ref{a28.2}] Fix $t\geq 0$ and $3\leq x \leq n-2$.
For a fixed $2\leq y\leq n-1$, $\P(x_t=y) = 1/(n-2)$. 
In view of \eqref{a28.1}, $\E v(y,t) =\mu(y,t)$ for all $y$.
Since $x_t$ could be $x-1, x $ and $x+1$ with equal probabilities, symmetry and Step \ref{a29.1} (see especially \eqref{a28.3})  show that
\begin{align*}
&\E  v_-(x,t+1) = \frac{n-5}{n-2} \E v(x,t)  \\
&\qquad + \P(x_t=x-1)  \frac 1 3 (\E v(x-2,t) + \E v(x-1,t) + \E v(x,t))
\\ 
&\qquad+ \P(x_t=x)  \frac 1 3 (\E v(x-1,t) + \E v(x,t) + \E v(x+1,t))\\
&\qquad+ \P(x_t=x+1)   \frac 1 3 (\E v(x,t) + \E v(x+1,t) + \E v(x+2,t)) \\ 
&= \frac{n-5}{n-2} \mu(x,t)  \\
&\qquad +\frac 1{n-2}  \cdot  \frac 1 3 (\mu(x-2,t) + \mu(x-1,t) + \mu(x,t))
\\ 
&\qquad+ \frac 1{n-2}  \cdot  \frac 1 3 (\mu(x-1,t) + \mu(x,t) + \mu(x+1,t))\\
&\qquad+ \frac 1{n-2}  \cdot  \frac 1 3 (\mu(x,t) + \mu(x+1,t) + \mu(x+2,t)) \\
&\quad=  \mu(x,t)-\frac{3}{n-2} \mu(x,t)  \\
&\qquad +\frac 1{n-2} \cdot \frac 1 3 (\mu(x-2,t) +2 \mu(x-1,t) +3 \mu(x,t) + 2\mu(x+1,t) + \mu(x+2,t)) \\
&\quad=  \mu(x,t)  \\
&\qquad +\frac 1{n-2} \cdot \frac 1 3 (\mu(x-2,t) +2 \mu(x-1,t) -6 \mu(x,t) + 2\mu(x+1,t) + \mu(x+2,t)) ,
\end{align*}
so
\begin{align}\notag
&(n-2)\E ( v_-(x,t+1) - v(x,t))=(n-2)\E ( v_-(x,t+1) - \mu(x,t))  \\
&\qquad= \frac 1 3 (\mu(x-2,t) +2 \mu(x-1,t) -6 \mu(x,t) + 2\mu(x+1,t) + \mu(x+2,t)) .\label{a29.3}
\end{align}

Our model, encapsulated in \eqref{a28.1}, implies that $\E v(y,t)^2 =\mu(y,t)^2 + \sigma(y,t)^2$ for all $y$. We use symmetry, Step \ref{a29.1} and \eqref{a28.4} to obtain the following formula analogous to \eqref{a29.3},
\begin{align}\label{m19.2}
(n-2)&\E (v_-(x,t+1)^2-v(x,t)^2) \\
&\quad= \frac 1 3 \Big(\sigma(x-2,t)^2 + \mu(x-2,t)^2 +2 (\sigma(x-1,t)^2 + \mu(x-1,t)^2) \notag \\
&\qquad -6( \sigma(x,t)^2+ \mu(x,t)^2) + 2(\sigma(x+1,t)^2+ \mu(x+1,t)^2) \notag \\
&\qquad + \sigma(x+2,t)^2+ \mu(x,t+1)^2\Big). \notag 
\end{align}

Let
\begin{align}\label{m12.7}
\wh\mu(x,t) &= (n-2)\E ( v(x,t+1) - v_-(x,t+1)),\\
\calE(x,t)&= (n-2)\E (v(x,t+1)^2-v_-(x,t+1)^2) .\label{m12.8}
\end{align}
Then, after adding $\wh{\mu}$ and $\calE$ to both sides of \eqref{a29.3} and \eqref{m19.2}, respectively, we obtain
\begin{align}\label{j23.3}
&\left(n-2\right)
\E (v(x, t+1)-v(x,t))\\
&\quad = \frac 1 3 (\mu(x-2,t) +2 \mu(x-1,t) -6 \mu(x,t) + 2\mu(x+1,t) + \mu(x+2,t)) + \wh\mu(x,t),\notag \\
&\quad = \frac 2 {(n-1)^2}  \wt\mu_{xx,n}(x,t)
 + O(n^{-4}) + \wh\mu(x,t),\notag \\
&\left(n-2\right)
\E (v(x, t+1)^2-v(x,t)^2)\label{j23.4}\\
 &\quad= \frac 1 3 \Big(\sigma(x-2,t)^2 + \mu(x-2,t)^2 +2 (\sigma(x-1,t)^2 + \mu(x-1,t)^2) \notag \\
&\qquad -6( \sigma(x,t)^2+ \mu(x,t)^2) + 2(\sigma(x+1,t)^2+ \mu(x+1,t)^2) \notag \\
&\qquad + \sigma(x+2,t)^2+ \mu(x,t+1)^2\Big)+\calE(x,t)\notag\\
 &\quad=  \frac 2 {(n-1)^2}
 \left( \wt\mu(x,t)^2 + \wt \sigma(x,t)^2\right)_{xx,n}
 + O(n^{-4})+\calE(x,t). \notag 
 \end{align}

Next, we observe that by Step \ref{a29.2},
\begin{align*}
& \wh\mu(x,t) \\
&= (n-2) \E ((v_-(x-1,t+1)  - v_-(x,t+1)) \bone_{ v_-(x-1,t+1) > v_-(x,t+1)}\mid x_t =x-1, \kappa_t=1)\\
&\qquad \times \P(x_t=x-1, \kappa_t=1) \\
& + (n-2) \E ((v_-(x-1,t+1)  - v_-(x,t+1)) \bone_{ v_-(x-1,t+1) > v_-(x,t+1)}\mid x_t=x, \kappa_t=-1)\\
&\qquad \times \P(x_t=x, \kappa_t=-1) \\
& + (n-2) \E ((v_-(x+1,t+1)  - v_-(x,t+1)) \bone_{v_-(x,t+1) > v_-(x+1,t+1)}\mid x_t=x,\kappa_t=1)\\
&\qquad \times \P(x_t=x, \kappa_t=1) \\
& + (n-2) \E ((v_-(x+1,t+1)  - v_-(x,t+1)) \bone_{v_-(x,t+1) > v_-(x+1,t+1)}\mid x_t=x+1,\kappa_t=-1)\\
&\qquad \times \P(x_t=x+1, \kappa_t=-1) \\
&= \frac 1 2 \E ((v_-(x-1,t+1)  - v_-(x,t+1)) \bone_{ v_-(x-1,t+1) > v_-(x,t+1)}\mid x_t =x-1, \kappa_t=1) \\
& +  \frac 1 2  \E ((v_-(x-1,t+1)  - v_-(x,t+1)) \bone_{ v_-(x-1,t+1) > v_-(x,t+1)}\mid x_t=x, \kappa_t=-1)\\
& + \frac 1 2  \E ((v_-(x+1,t+1)  - v_-(x,t+1)) \bone_{v_-(x,t+1) > v_-(x+1,t+1)}\mid x_t=x,\kappa_t=1)\\
& + \frac 1 2  \E ((v_-(x+1,t+1)  - v_-(x,t+1)) \bone_{v_-(x,t+1) > v_-(x+1,t+1)}\mid x_t=x+1,\kappa_t=-1).
\end{align*}
If $x_t=x$ then the sequence $(v_-(x-1,t+1) , v_-(x,t+1),v_-(x+1,t+1))$ is exchangeable---this follows from the definition given in Step \ref{a29.1}. Hence, the two middle terms in the above formula cancel each other and we obtain
\begin{align}\label{a30.1}
& \wh\mu(x,t) \\
&= \frac 1 2 \E ((v_-(x-1,t+1)  - v_-(x,t+1)) \bone_{ v_-(x-1,t+1) > v_-(x,t+1)}\mid x_t =x-1, \kappa_t=1)\notag \\
& + \frac 1 2  \E ((v_-(x+1,t+1)  - v_-(x,t+1)) \bone_{v_-(x,t+1) > v_-(x+1,t+1)}\mid x_t=x+1,\kappa_t=-1).\notag
\end{align}

A similar calculation yields
\begin{align}\label{a30.2}
& \calE(x,t) \\
&= \frac 1 2 \E \left(\left(v_-(x-1,t+1)^2  - v_-(x,t+1)^2\right) \bone_{ v_-(x-1,t+1) > v_-(x,t+1)}\mid x_t =x-1, \kappa_t=1\right) \notag\\
& + \frac 1 2  \E \left(\left(v_-(x+1,t+1)^2  - v_-(x,t+1)^2\right) \bone_{v_-(x,t+1) > v_-(x+1,t+1)}\mid x_t=x+1,\kappa_t=-1\right).\notag
\end{align}

Define $a(x,t)$ and $r(x,t)\geq 0$ by
\begin{align}\label{d5.2}
a(x,t) &=\frac 1 3 (v(x-1,t) + v(x,t) + v(x+1,t)),\\
r(x,t)^2 &=\frac 2 3 \left((v(x-1,t)-a(x,t))^2 + (v(x,t)-a(x,t))^2+(v(x+1,t)-a(x,t))^2\right) \label{d5.3}\\
&= \frac 4 9 \big( v(x-1,t)^2 + v(x,t)^2 + v(x+1,t)^2 \label{d5.1}\\
&\qquad- v(x-1,t)  v(x,t) - v(x-1,t)  v(x+1,t) - v(x,t)  v(x+1,t)\big). \notag
\end{align}

If $x_t=x$ then
the vector $(v_-(x-1,t+1) , v_-(x,t+1) , v_-(x+1,t+1))$ is distributed uniformly on the circle in $\R^3$ given by the parametric formula,
\begin{align*}
\left(a(x,t) + r(x,t)\sin \theta,
a(x,t) + r(x,t)\sin( \theta+2\pi/3),
a(x,t) +r(x,t) \sin( \theta+4\pi/3)\right), 
\end{align*}
for $\theta\in[0,2\pi)$.

Let $\calF_{x,t} $ denote the $\sigma$-field generated by $v(x-1,t) , v(x,t) $ and $ v(x+1,t)$.
We have
\begin{align}\label{oc24.3}
\E &((v_-(x-1,t+1)  - v_-(x,t+1)) \bone_{ v_-(x-1,t+1) > v_-(x,t+1)}\mid x_t=x-1, \kappa_t =1, \calF_{x-1,t})\\
&= \frac 1{2\pi}\int_{-\pi/2}^{\pi/2} \Big((a(x-1,t) + r(x-1,t)\sin( \theta+2\pi/3))\notag \\
&\qquad - ( a(x-1,t) + r(x-1,t)\sin( \theta+4\pi/3))\Big) d\theta\notag\\
&= \frac 1{2\pi} \int_{-\pi/2}^{\pi/2}  r(x-1,t)
(\sin( \theta+2\pi/3) - \sin( \theta+4\pi/3))d\theta
 = \frac{\sqrt{3}}\pi  r(x-1,t).\notag
\end{align}
The following formula is analogous,
\begin{align*}
\E &((v_-(x+1,t+1)  - v_-(x,t+1)) \bone_{v_-(x,t+1) > v_-(x+1,t+1)}\mid x_t=x+1,\kappa_t=-1,
\calF_{x+1,t})\\
&= -\frac{\sqrt{3}}\pi  r(x+1,t).
\end{align*}
Thus, \eqref{a30.1} and \eqref{oc24.3} imply that
\begin{align}\label{oc24.5}
\wh\mu(x,t) = \frac {\sqrt{3}} {2\pi}
\E( r(x-1,t) - r(x+1,t)).
\end{align}

Similarly, we see that
\begin{align}\label{oc24.4}
\E &((v_-(x-1,t+1)^2 - v_-(x,t+1)^2) \bone_{ v_-(x-1,t+1) > v_-(x,t+1)}\mid x_t=x-1, \kappa_t =1, \calF_{x-1,t})\\
&= \frac 1{2\pi} \int_{-\pi/2}^{\pi/2} \Big((a(x-1,t) + r(x-1,t)\sin( \theta+2\pi/3))^2 
\notag \\
&\qquad - ( a(x-1,t) + r(x-1,t)\sin( \theta+4\pi/3))^2\Big) d\theta\notag \\
&=  \frac 1{2\pi}\int_{-\pi/2}^{\pi/2} 
2 a(x-1,t)  r(x-1,t)(\sin( \theta+2\pi/3) - \sin( \theta+4\pi/3))d\theta\notag \\
&\quad +  \frac 1{2\pi} \int_{-\pi/2}^{\pi/2} r(x-1,t)^2(\sin^2( \theta+2\pi/3) - \sin^2( \theta+4\pi/3)))d\theta\notag \\
& = \frac 1{2\pi}\cdot 2 a(x-1,t)  r(x-1,t) 2\sqrt{3}
+ r(x-1,t)^2 \cdot 0\notag \\
& = \frac{2\sqrt{3}}\pi a(x-1,t)  r(x-1,t).\notag 
\end{align}
Additionally, an analogous calculation to the previous one yields
\begin{align*}
\E &\left(\left(v_-(x+1,t+1)^2  - v_-(x,t+1)^2\right) \bone_{v_-(x,t+1) > v_-(x+1,t+1)}\mid x_t=x+1,\kappa_t=-1,
\calF_{x+1,t}\right)\\
&= -\frac{2\sqrt{3}}\pi a(x+1,t)  r(x+1,t).\notag 
\end{align*}
Hence we see that \eqref{a30.2} and \eqref{oc24.4} together imply that
\begin{align}\label{oc24.6}
\calE (x,t) &= \frac {\sqrt{3}} {\pi}
\E(a(x-1,t) r(x-1,t) - a(x+1,t)r(x+1,t)).
\end{align}

Recall that
the density of a normal random variable with mean $\alpha$ and variance $\beta^2$ is
\begin{align*}
f_{\alpha,\beta}(u)=
\frac 1 {\sqrt{2\pi} \beta} \exp\left ( - \frac {(u-\alpha)^2}{2 \beta^2}\right),
\end{align*}
and the joint density of three independent normal random variables with parameters $\bal=(\alpha_1,\alpha_2,\alpha_3)$ and $\bbet=(\beta_1,\beta_2,\beta_3)$ is
\begin{align*}
f_{\bal,\bbet}(u,y,z)=
\frac 1 {(2\pi)^{3/2} \beta_1\beta_2\beta_3} 
\exp\left ( - \frac {(u-\alpha_1)^2}{2 \beta_1^2}
- \frac {(y-\alpha_2)^2}{2 \beta_2^2}
- \frac {(z-\alpha_3)^2}{2 \beta_3^2}\right).
\end{align*}
It follows from \eqref{d5.1} that
\begin{align}\label{j1.1}
\E r(x,t) =
\int_{\R^3}
\frac 2 3 
\sqrt{u^2+y^2+z^2 - uy - yz-uz}f_{\bal,\bbet}(u,y,z) \,du \,dy \,dz,
\end{align}
where
\begin{align}\label{oc24.7}
\bal &= (\mu(x-1,t) , \mu(x,t) , \mu(x+1,t)),\\
\bbet &= (\sigma(x-1,t) , \sigma(x,t) , \sigma(x+1,t)).\label{oc24.8}
\end{align}

Let
\begin{align}
\bal_- &= (\mu(x-2,t) , \mu(x-1,t) , \mu(x,t)),\label{oc24.9}\\
\bbet_- &= (\sigma(x-2,t) , \sigma(x-1,t) , \sigma(x,t)),\label{oc24.10}\\
\bal_+ &= (\mu(x,t) , \mu(x+1,t) , \mu(x+2,t)),\label{oc24.11}\\
\bbet_+ &= (\sigma(x,t) , \sigma(x+1,t) , \sigma(x+2,t)).\label{oc24.12}
\end{align}
According to \eqref{oc24.5},
\begin{align}\notag
\wh\mu(x,t) &= \frac {\sqrt{3}} {2\pi}
\E( r(x-1,t) - r(x+1,t))\\
&= \frac {\sqrt{3}} {2\pi}
\int_{\R^3}
\frac 2 3 
\sqrt{u^2+y^2+z^2 - uy - yz-uz}f_{\bal_-,\bbet_-}(u,y,z) \,du\, dy\, dz\label{oc24.13} \\
&\quad -\frac {\sqrt{3}}{2\pi}
\int_{\R^3}
\frac 2 3 
\sqrt{u^2+y^2+z^2 - uy - yz-uz}f_{\bal_+,\bbet_+}(u,y,z)\, du\, dy\, dz.\notag
\end{align}

According to \eqref{oc24.6},
\begin{align}\label{m1.1}
\calE (x,t) &= \frac {\sqrt{3}} {\pi}
\E(a(x-1,t-1) r(x-1,t-1) - a(x+1,t-1)r(x+1,t-1))\\
&= \frac {\sqrt{3}} {\pi}
\int_{\R^3} \frac 1 3 (u+y+z)
\frac 2 3 
\sqrt{u^2+y^2+z^2 - uy - yz-uz}f_{\bal_-,\bbet_-}(u,y,z) \,du\, dy\, dz\notag \\
&\quad -\frac {\sqrt{3}}{\pi}
\int_{\R^3} \frac 1 3 (u+y+z)
\frac 2 3 
\sqrt{u^2+y^2+z^2 - uy - yz-uz}f_{\bal_+,\bbet_+}(u,y,z) \,du \,dy\, dz.\notag
\end{align}

In view of \eqref{oc24.13} and \eqref{m1.1}, we need to compute (or at least estimate) integrals of the form
\begin{align}\label{m12.1}
I_1(\bal,\bbet)& :=
\int_{\R^3}
\sqrt{\hat u^2+\hat y^2+\hat z^2 - \hat u\hat y - \hat y\hat z-\hat u\hat z}f_{\bal,\bbet}(\hat u,\hat y,\hat z) d\hat u d\hat y d\hat z,\\
I_2(\bal,\bbet)& :=
\int_{\R^3} (\hat u+\hat y+\hat z)
\sqrt{\hat u^2+\hat y^2+\hat z^2 - \hat u\hat y - \hat y\hat z-\hat u\hat z}f_{\bal,\bbet}(\hat u,\hat y,\hat z) d\hat u d\hat y d\hat z.\label{m12.2}
\end{align}

Given $\alpha_1,\alpha_2,\alpha_3$ and $\eps>0$, let $\alpha =(\alpha_1+\alpha_2+\alpha_3)/3$. It is elementary to check that there exist unique $\delta_1$ and $\gamma_1$ such that
\begin{align}\label{j1.6n}
\bal=(\alpha_1,\alpha_2,\alpha_3) = (\alpha-\eps \delta_1+\eps^2 \gamma_1, \alpha-2\eps^2 \gamma_1 , \alpha + \eps \delta_1+\eps^2 \gamma_1).
\end{align}
The following formula is analogous,
\begin{align}
\bbet=(\beta_1,\beta_2,\beta_3) = (\beta-\eps \delta_2 + \eps^2\gamma_2, \beta -2\eps^2\gamma_2, \beta +   \eps \delta_2+\eps^2 \gamma_2),\label{j1.7n_jgh}
\end{align}
with $\beta$ denoting the average of  $\bbet$. We will assume that $\beta_k>0$ for $k=1,2,3$.

According to Lemma \ref{j31.3},
\begin{align}\label{j23.1}
I_1(\bal,\bbet)& =\frac{\sqrt{3\pi}\beta}{2} \left[1+\frac{\eps^2}{2\beta^2}\left(\delta_1^2+\frac{\delta_2^2}{2} \right) \right] + O(\eps^3),\\
\label{j23.2}
I_2(\bal,\bbet)& = 3\alpha \frac{\sqrt{{3\pi\beta^2}}}{2} \left[1+\frac{\eps^2}{2\beta^2}\left(\delta_1^2+\frac{\delta_2^2}{2} \right) \right] +	2\eps^2\sqrt{3\pi}\delta_1\delta_2 +O(\eps^3).
\end{align}

Recall \eqref{oc24.9}-\eqref{oc24.12} and assume that
\begin{align}\label{m17.2}
\bal_- &= (\mu(x-2,t) , \mu(x-1,t) , \mu(x,t))\\
&= \left(\mu(x,t)-2\eps\delta_1 , \mu(x,t)-\eps\delta_1-3\eps^2\gamma_1 , \mu(x,t)\right)
, \notag\\
\bbet_- &= (\sigma(x-2,t) , \sigma(x-1,t) , \sigma(x,t))\label{m17.3}\\
&=
\left(\sigma(x,t)-2\eps\delta_2  , \sigma(x,t)-\eps\delta_2-3\eps^2\gamma_2 , \sigma(x,t)\right),\notag\\
\bal_+ &= (\mu(x,t) , \mu(x+1,t) , \mu(x+2,t))\label{m17.4}\\
&= \left(\mu(x,t), \mu(x,t)+\eps\delta_3 -3\eps^2\gamma_3, \mu(x,t)+2\eps\delta_3 \right),\notag\\
\bbet_+ &= (\sigma(x,t) , \sigma(x+1,t) , \sigma(x+2,t))\label{m17.5}\\
&=
\left(\sigma(x,t),\sigma(x,t)+\eps\delta_4 -3\eps^2\gamma_4 , \sigma(x,t)+2\eps\delta_4\right).\notag
\end{align}

Then \eqref{j23.1} implies that
\begin{align*}
     I_1&(\bal_-,\bbet_-)\\
     & =\frac{(\sigma(x,t) -\eps\delta_2 -\eps^2 \gamma_2)\sqrt{ {3\pi}}}{2} +
     \frac{\sqrt{ {3\pi}}}{2}
     \frac{\eps^2}{2(\sigma(x,t) -\eps\delta_2 -\eps^2 \gamma_2)}\left(\delta_1^2+\frac{\delta_2^2}{2} \right)  + O(\eps^3),\\ 
     I_1&(\bal_+,\bbet_+)\\
     & =\frac{(\sigma(x,t) +\eps\delta_4 -\eps^2 \gamma_4)\sqrt{ {3\pi}}}{2} +
     \frac{\sqrt{ {3\pi}}}{2}
     \frac{\eps^2}{2(\sigma(x,t) +\eps\delta_4 -\eps^2 \gamma_4)}\left(\delta_3^2+\frac{\delta_4^2}{2} \right)  + O(\eps^3).
\end{align*}

Recall \eqref{oc24.13} and \eqref{m12.1} to see that 
\begin{align}\notag
\wh\mu(x,t) &= \frac {1} {\sqrt{3}\pi}
\int_{\R^3}
\sqrt{u^2+y^2+z^2 - uy - yz-uz}f_{\bal_-,\bbet_-}(u,y,z) du dy dz \\
&\quad -\frac {1} {\sqrt{3}\pi}
\int_{\R^3}
\sqrt{u^2+y^2+z^2 - uy - yz-uz}f_{\bal_+,\bbet_+}(u,y,z) du dy dz\notag\\
&= \frac {1} {\sqrt{3}\pi} (I_1(\bal_-,\bbet_-) - I_1(\bal_+,\bbet_+)) \notag\\
 & =\frac{\sigma(x,t) -\eps\delta_2 -\eps^2 \gamma_2}{2\sqrt{\pi}} +
     \frac{1}{2\sqrt{\pi}}
     \frac{\eps^2}{2(\sigma(x,t) -\eps\delta_2 -\eps^2 \gamma_2)}\left(\delta_1^2+\frac{\delta_2^2}{2} \right)  + O(\eps^3) \notag\\ 
     & -\frac{\sigma(x,t) +\eps\delta_4 -\eps^2 \gamma_4}{2\sqrt{\pi}} -
     \frac{1}{2\sqrt{\pi}}
     \frac{\eps^2}{2(\sigma(x,t) +\eps\delta_4 -\eps^2 \gamma_4)}\left(\delta_3^2+\frac{\delta_4^2}{2} \right)  + O(\eps^3) \notag\\
     &= -\frac{1}{2\sqrt{\pi}} \eps (\delta_2 +\delta_4)
     + \eps^2 \frac{1}{2\sqrt{\pi}}
     \left(-\gamma_2 + \gamma_4  \right)  + O(\eps^3)
     .\label{j19.2}
\end{align}
The terms involving $\delta_1^2$ and $\delta_3^2$ canceled each other in the above calculation.
More precisely, they combined into an $O(\eps^3)$ quantity because, essentially,
$\delta_1$ and $\delta_3$ are derivatives at the nearby locations; see \eqref{j23.10}-\eqref{j23.13}
below for more details. A similar remark applies to the cancellation of terms 
involving $\delta_2^2$ and $\delta_4^2$.

In the next part of the proof, it will be convenient to use the notation $\eps = 1/(n-1)$.

We relate $\delta_k$'s and $\gamma_k$'s to derivatives of $\wt \mu$ and $\wt \sigma$ as follows.
\begin{align}\label{j23.10}
2 \eps \delta_1 &= \mu(x,t) - (\mu(x,t) - 2 \eps \delta_1)
= \mu(x,t) - \mu(x-2,t) \\
&= \frac 2 {n-1} \wt \mu_{x,n}(x,t)
- \frac{1}{2} \cdot \frac 4 {(n-1)^2} \wt \mu_{xx,n}(x,t)
+ O(n^{-3}).\notag
\end{align}
Similarly,
\begin{align}\label{j23.11}
2 \eps \delta_3 &=  (\mu(x,t) + 2 \eps \delta_3)-\mu(x,t)
= \mu(x+2,t) - \mu(x,t) \\
&= \frac 2 {n-1} \wt \mu_{x,n}(x,t)
+\frac{1}{2} \cdot  \frac 4 {(n-1)^2} \wt \mu_{xx,n}(x,t)
+ O(n^{-3}),\notag\\
3\eps^2 \gamma_1 &=
\frac 1 {2(n-1)^2} \wt \mu_{xx,n}(x,t)
+ O(n^{-3}),\label{j23.12}\\
3\eps^2 \gamma_3 &=
\frac 1 {2(n-1)^2} \wt \mu_{xx,n}(x,t)
+ O(n^{-3}).\label{j23.13}
\end{align}
Analogous formulas hold for $\delta_k$ and $\gamma_k$ for $k=2,4$.
This and \eqref{j19.2} imply that
\begin{align}\notag
\wh\mu(x,t) &=  -\frac{1}{2\sqrt{\pi}} \eps (\delta_2 +\delta_4)
     + \eps^2 \frac{1}{2\sqrt{\pi}}
     \left(-\gamma_2 + \gamma_4  \right)  + O(\eps^3)\\
& = -\frac{1}{\sqrt{\pi}(n-1)}
\wt \sigma_{x,n}(x,t)+ O(n^{-3})
    .\label{j23.5}
\end{align}

Next we use \eqref{j23.2} to obtain 
\begin{align*}
     I_2(\bal_-,\bbet_-)
     & =\frac{(\mu(x,t) -\eps\delta_1 -\eps^2 \gamma_1)
     (\sigma(x,t) -\eps\delta_2 -\eps^2 \gamma_2)3\sqrt{ {3\pi}}}{2}\\
     &\quad +
     \frac{3\sqrt{ {3\pi}}}{2}
     \frac{\eps^2 (\mu(x,t) -\eps\delta_1 -\eps^2 \gamma_1)}{2(\sigma(x,t) -\eps\delta_2 -\eps^2 \gamma_2)}\left(\delta_1^2+\frac{\delta_2^2}{2} \right) +	2\eps^2\sqrt{3\pi}\delta_1\delta_2  + O(\eps^3),\\ 
     I_2(\bal_+,\bbet_+)
     & =\frac{(\mu(x,t) +\eps\delta_3 -\eps^2 \gamma_3)
     (\sigma(x,t) +\eps\delta_4 -\eps^2 \gamma_4)3\sqrt{ {3\pi}}}{2} \\
     &\quad +
     \frac{3\sqrt{ {3\pi}}}{2}
     \frac{\eps^2 (\mu(x,t) +\eps\delta_3 -\eps^2 \gamma_3)}{2(\sigma(x,t) -\eps\delta_4 -\eps^2 \gamma_4)}\left(\delta_3^2+\frac{\delta_4^2}{2} \right) +	2\eps^2\sqrt{3\pi}\delta_3\delta_4 + O(\eps^3).
\end{align*}

Recall  \eqref{m1.1} and \eqref{m12.2} to see that,
\begin{align*}\notag
\calE (x,t) 
&= \frac {\sqrt{3}} {\pi}
\int_{\R^3} \frac 1 3 (u+y+z)
\frac 2 3 
\sqrt{u^2+y^2+z^2 - uy - yz-uz}f_{\bal_-,\bbet_-}(u,y,z) du dy dz \\
&\quad -\frac {\sqrt{3}}{\pi}
\int_{\R^3} \frac 1 3 (u+y+z)
\frac 2 3 
\sqrt{u^2+y^2+z^2 - uy - yz-uz}f_{\bal_+,\bbet_+}(u,y,z) du dy dz\notag\\
&= \frac{ 2 \sqrt{3}}{3\pi} \cdot
\frac 1 3  I_2(\bal_-,\bbet_-) -\frac{ 2 \sqrt{3}}{3\pi} \cdot
 \frac 1 3  I_2(\bal_+,\bbet_+) \\
     & = \frac{ 2 \sqrt{3}}{3\pi} \cdot
\frac 1 3 \cdot \frac{3\sqrt{ {3\pi}}}{2}
     (\mu(x,t) -\eps\delta_1 -\eps^2 \gamma_1)
     (\sigma(x,t) -\eps\delta_2 -\eps^2 \gamma_2)\\
     &\quad + \frac{ 2 \sqrt{3}}{3\pi} \cdot
\frac 1 3 \cdot 
     \frac{3\sqrt{ {3\pi}}}{2}
     \frac{\eps^2 (\mu(x,t) -\eps\delta_1 -\eps^2 \gamma_1)}{2(\sigma(x,t) -\eps\delta_2 -\eps^2 \gamma_2)}\left(\delta_1^2+\frac{\delta_2^2}{2} \right) + \frac{ 2 \sqrt{3}}{3\pi} \cdot
\frac 1 3 \cdot 	\eps^2\sqrt{3\pi}\delta_1\delta_2 \\ 
     & \quad - \frac{ 2 \sqrt{3}}{3\pi} \cdot
\frac 1 3 \cdot \frac{3\sqrt{ {3\pi}}}{2} (\mu(x,t) +\eps\delta_3 -\eps^2 \gamma_3)
     (\sigma(x,t) +\eps\delta_4 -\eps^2 \gamma_4) \\
     &\quad -
      \frac{ 2 \sqrt{3}}{3\pi} \cdot
\frac 1 3 \cdot  \frac{3\sqrt{ {3\pi}}}{2}
     \frac{\eps^2 (\mu(x,t) +\eps\delta_3 -\eps^2 \gamma_3)}{2(\sigma(x,t) -\eps\delta_4 -\eps^2 \gamma_4)}\left(\delta_3^2+\frac{\delta_4^2}{2} \right) -
    \frac{ 2 \sqrt{3}}{3\pi} \cdot
\frac 1 3 \cdot 	\eps^2\sqrt{3\pi}\delta_3\delta_4\\
&\qquad + O(\eps^3 )\\
     & = \frac{ 1}{\sqrt{ \pi}}
     (\mu(x,t) -\eps\delta_1 -\eps^2 \gamma_1)
     (\sigma(x,t) -\eps\delta_2 -\eps^2 \gamma_2)\\
     &\quad + \frac{ 1}{\sqrt{ \pi}}
     \frac{\eps^2 (\mu(x,t) -\eps\delta_1 -\eps^2 \gamma_1)}{2(\sigma(x,t) -\eps\delta_2 -\eps^2 \gamma_2)}\left(\delta_1^2+\frac{\delta_2^2}{2} \right) \\ 
     & \quad - \frac{ 1}{\sqrt{ \pi}} (\mu(x,t) +\eps\delta_3 -\eps^2 \gamma_3)
     (\sigma(x,t) +\eps\delta_4 -\eps^2 \gamma_4) \\
     &\quad -
      \frac{ 1}{\sqrt{ \pi}}
     \frac{\eps^2 (\mu(x,t) +\eps\delta_3 -\eps^2 \gamma_3)}{2(\sigma(x,t) -\eps\delta_4 -\eps^2 \gamma_4)}\left(\delta_3^2+\frac{\delta_4^2}{2} \right) 
     +\frac{2}{3\sqrt{\pi}}\eps^2 ( \delta_1 \delta_2 - \delta_3\delta_4) + O(\eps^3 )\\
& = \frac{ 1}{\sqrt{ \pi}}
\eps(-(\delta_1+\delta_3) \sigma(x,t) -(\delta_2+\delta_4) \mu(x,t))\\
&\quad + \frac{ 1}{\sqrt{ \pi}}
\eps^2 ( - \gamma_1 \sigma(x,t) - \gamma_2 \mu(x,t) + \gamma_3 \sigma(x,t) + \gamma_4 \mu(x,t)) \\
&\quad+\frac{5}{3\sqrt{\pi}}\eps^2 ( \delta_1 \delta_2 - \delta_3\delta_4) + O(\eps^3 ).
\end{align*}
It follows from \eqref{j23.10}-\eqref{j23.13} that
\begin{align*}
 \frac{ 1}{\sqrt{ \pi}}
\eps^2 ( - \gamma_1 \sigma(x,t) - \gamma_2 \mu(x,t) + \gamma_3 \sigma(x,t) + \gamma_4 \mu(x,t)) +\frac{5}{3\sqrt{\pi}}\eps^2 ( \delta_1 \delta_2 - \delta_3\delta_4)= O(\eps^3 ),
\end{align*}
so
\begin{align*}\notag
\calE (x,t) 
&= \frac{ 1}{\sqrt{ \pi}}
\eps(-(\delta_1+\delta_3) \sigma(x,t) -(\delta_2+\delta_4) \mu(x,t)) + O(\eps^3 )\\
& =-\frac{ 2}{\sqrt{ \pi}(n-1)}\left( \sigma(x,t)
 \wt \mu_{x,n}(x,t)
- \mu(x,t)
\wt \sigma_{x,n}(x,t)
\right) + O(n^{-3})
.
\end{align*}
By combining this estimate with \eqref{j23.4} and similarly combining 
\eqref{j23.5} with \eqref{j23.3}, we obtain \eqref{m12.10}-\eqref{m12.11}.
\end{proof}

\subsubsection{Estimates for $I_1$ and $I_2$}

Recall integrals defined in \eqref{m12.1}-\eqref{m12.2}:
\begin{align}\label{m12.1a}
I_1(\bal,\bbet)& =
\int_{\R^3}
\sqrt{\hat u^2+\hat y^2+\hat z^2 - \hat u\hat y - \hat y\hat z-\hat u\hat z}f_{\bal,\bbet}(\hat u,\hat y,\hat z) d\hat u d\hat y d\hat z,\\
I_2(\bal,\bbet)& =
\int_{\R^3} (\hat u+\hat y+\hat z)
\sqrt{\hat u^2+\hat y^2+\hat z^2 - \hat u\hat y - \hat y\hat z-\hat u\hat z}f_{\bal,\bbet}(\hat u,\hat y,\hat z) d\hat u d\hat y d\hat z.\label{m12.2a}
\end{align}
We will provide estimates for the two integrals needed in the proof of Theorem \ref{a28.2}.

Recall the following notation from \eqref{j1.6n}-\eqref{j1.7n_jgh}.
Given $\alpha_1,\alpha_2,\alpha_3$ and $\eps>0$, let $\alpha =(\alpha_1+\alpha_2+\alpha_3)/3$,
\begin{align}\label{j1.6na}
\bal&=(\alpha_1,\alpha_2,\alpha_3) = (\alpha-\eps \delta_1+\eps^2 \gamma_1, \alpha-2\eps^2 \gamma_1 , \alpha + \eps \delta_1+\eps^2 \gamma_1),\\
\bbet&=(\beta_1,\beta_2,\beta_3) = (\beta-\eps \delta_2 + \eps^2\gamma_2, \beta -2\eps^2\gamma_2, \beta +   \eps \delta_2+\eps^2 \gamma_2),\label{j1.7n_jgha}
\end{align}
with $\beta$ denoting the average of  $\bbet$. Assume that $\beta_k>0$ for $k=1,2,3$.

\begin{lemma}\label{j31.3}
We have 
\begin{align}\label{j23.1a}
I_1(\bal,\bbet) &=\frac{\sqrt{3\pi}\beta}{2} \left[1+\frac{\eps^2}{2\beta^2}\left(\delta_1^2+\frac{\delta_2^2}{2} \right) \right] + O(\eps^3),\\
\label{j23.2a}
I_2(\bal,\bbet) &= 3\alpha \frac{\sqrt{{3\pi\beta^2}}}{2} \left[1+\frac{\eps^2}{2\beta^2}\left(\delta_1^2+\frac{\delta_2^2}{2} \right) \right] +	\eps^2\sqrt{3\pi}\delta_1\delta_2 +O(\eps^3).
\end{align}
\end{lemma}

\begin{proof}
We make the change of variables $u=(\hat{u}-\alpha)/\beta,$ $y=(\hat{y}-\alpha)/\beta,$ $z= (\hat{z}-\alpha)/\beta.$
Then,
\begin{align*}
&(\hat{u}-\alpha_1)^2/\beta_1^2 = u^2 + 2u\eps \frac{\delta_1 +\delta_2 u}{\beta}+\frac{\eps^2}{\beta^2}\left[(3\delta_2^2-2\beta \gamma_2)u^2-(2\beta \gamma_1-4\delta_1\delta_2)u+\delta_1^2 \right]\\ 
&\quad+(1+u^2)O(\eps^3),\\
&(\hat{y}-\alpha_2)^2/\beta_2^2 = y^2+\frac{4 \eps^2}{\beta}\left[\gamma_2 y^2+\gamma_1y\right]+(1+y^2)O(\eps^4),\\
&(\hat{z}-\alpha_3)^2/\beta_3^2 = z^2 - 2z\eps \frac{\delta_1 +\delta_2 z}{\beta}+\frac{\eps^2}{\beta^2}\left[(3\delta_2^2-2\beta \gamma_2)z^2-(2\beta \gamma_1 - 4\delta_1\delta_2)z+\delta_1^2 \right]\\ 
&\quad +(1+z^2)O(\eps^3)  .
\end{align*}
We make another change of variables,
$$s=(z-u)/\sqrt{2},\quad\quad b = (y-(u+z)/2) \sqrt{2/3},\quad\quad {\rm and}\quad \quad w=(u+y+z)/\sqrt{3},$$
and we note that $s,b$ and $w$ form an orthonormal coordinate system if $u,y$ and $z$ do. Hence the Jacobian of the transformation $(u,y,z) \to (s,b,w)$ is equal to 1.

The following identities are easy to verify,
\begin{align}\label{j31.2}
b^2 + s^2
=  \frac 2 3 ( u^2+  y^2+  z^2 -   u  y -   y  z-  u  z)
= \frac 2 {3\beta^2} (\hat u^2+\hat y^2+\hat z^2 - \hat u\hat y - \hat y\hat z-\hat u\hat z) .
\end{align}

Tedious but straightforward calculations show that
\begin{align}\label{j31.1}
&\frac{(\hat{u}-\alpha_1)^2}{\beta_1^2}+\frac{(\hat{y}-\alpha_2)^2}{\beta_2^2}+\frac{(\hat{z}-\alpha_3)^2}{\beta_3^2}\\
&=s^2+b^2+w^2-2\eps\left(\frac{\delta_1\sqrt{2} s}{\beta}+\sqrt{\frac{8}{3}}\frac{\delta_2 s(w-b/\sqrt{2})}{\beta}\right) \notag\\
&\quad+\frac{\eps^2}{\beta^2}\left[2\delta_1^2 +\frac{8}{\sqrt{3}}\delta_1\delta_2\left(w-\frac{b}{\sqrt{2}}\right)\right.\notag\\\vspace{1.4cm}\notag\\
&\quad\left. +\,\sqrt{24} \gamma_1 \beta  b +(3\delta_2^2-2\beta\gamma_2)(s^2+b^2+w^2)+(2\beta\gamma_2-\delta_2^2)(\sqrt{2}b+w)^2 \right]\notag\\\vspace{1.4cm}\notag\\
&\quad+(1+s^2+b^2+w^2)O(\eps^3).\notag
\end{align}

For ease of exposition, we define
\begin{align*}
\omega ^2 &:= s^2+b^2+w^2,\\
p_0 &:= w-\frac{b}{\sqrt{2}},\\
p_1 &:= \frac{\delta_1\sqrt{2} }{\beta}+\sqrt{\frac{8}{3}}\frac{\delta_2 p_0}{\beta},
\end{align*}
and
$$p_2:=\frac{1}{\beta^2}\left[2\delta_1^2 +\frac{8}{\sqrt{3}}\delta_1\delta_2 p_0 +\sqrt{24} \gamma_1 \beta  b +(3\delta_2^2-2\beta\gamma_2)\omega^2+(2\beta\gamma_2-\delta_2^2)(\sqrt{2}b+w)^2 \right]. $$
Then \eqref{j31.1} takes the form
\begin{align*}
&\frac{(\hat{u}-\alpha_1)^2}{\beta_1^2}+\frac{(\hat{y}-\alpha_2)^2}{\beta_2^2}+\frac{(\hat{z}-\alpha_3)^2}{\beta_3^2}
= \omega^2+2\eps\, s\, p_1 - \eps^2 p_2-\eps^3 R, 
\end{align*}
where the remainder satisfies $\eps^3 R=\eps^3 R(\eps) = (1+s^2+b^2+w^2)O(\eps^3)$. 
We combine this formula with \eqref{m12.1}, \eqref{j31.2} and the formula for the normal density to obtain
\begin{align}\label{eqn:i1_unexpand}
I_1&(\bal,\bbet)= \sqrt{\frac{3}{2}}\frac{\beta^4}{(2\pi)^\frac{3}{2}\beta_1\beta_2\beta_3} \\
&\quad\quad \times\int_{\R^3} \sqrt{b^2+s^2}\, \exp\left(-\omega^2/2+\eps\, s\, p_1 - (\eps^2/2) p_2-(\eps^3/2 )R \right)\,{\rm d}s\,{\rm d}b\,{\rm d}w.\nonumber
\end{align}

We apply Lemma \ref{j31.4} to obtain
\begin{align}\notag
I_1(\bal,\bbet) &= \sqrt{\frac{3}{2}}\frac{\beta^4}{(2\pi)^\frac{3}{2}\beta_1\beta_2\beta_3}\\
&\quad \times \int_{\R^3}(1+\eps s p_1-(\eps^2/2) (p_2-p_1^2s^2)+\tilde{R}) \sqrt{b^2+s^2}\, \exp\left(-\omega^2/2\right)\,{\rm d}s\,{\rm d}b\,{\rm d}w,\label{a1.1}
\end{align}
where $\tilde{R}=O(\eps^3)$.

The term in the integrand in \eqref{a1.1} which is linear in $\eps$ is an odd function of $s$. By symmetry it integrates to zero. Similar reasoning applies to the mixed terms in the quadratic term. Specifically, we can eliminate terms in $p_2$ and $p_1^2$ that contain  $b,w$ or $bw$. Thus,
\begin{align}\label{j13.2}
    I_1(\bal,\bbet) &= \sqrt{\frac{3}{2}}\frac{\beta^4}{(2\pi)^\frac{3}{2}\beta_1\beta_2\beta_3} \\
&\quad \times \int_{\R^3}(1-\frac{\eps^2}{2}(\widehat{p}_2-\widehat{p}_1 s^2)+\tilde{R}) \sqrt{b^2+s^2}\, \exp\left(-\omega^2/2\right)\,{\rm d}s\,{\rm d}b\,{\rm d}w,\notag
\end{align}
where
$$\widehat{p}_1 := 2\frac{\delta_1^2}{\beta^2} + 8\delta_2\frac{w^2+b^2/2}{3 \beta^2}$$
and
$$\widehat{p}_2 := \frac{1}{\beta^2} \left[2\delta_1^2+(3\delta_2^2-2 \beta \gamma_2) \omega^2 -(\delta_2^2-2\beta \gamma_2)(2b^2+w^2) \right].$$

Now, we consider the integral
\begin{align}\label{j13.1}
    P(q_1,q_2,q_3):=\int_{\R^3}\sqrt{b^2+s^2} \exp\left(-\frac{s^2+b^2+w^2}{2}-q_1s-q_2b-q_3w\right)\,{\rm d}s\,{\rm d}b\,{\rm d}w.
\end{align}
Performing the integral over $w,$ we see that
$$\frac{P(q_1,q_2,q_3)}{\sqrt{2\pi}}=\exp(q_3^2/2)\int_{\R^2}\sqrt{b^2+s^2} \exp\left(-\frac{s^2+b^2}{2}-q_1s-q_2b\right)\,{\rm d}s\,{\rm d}b.$$
Going into polar coordinates yields
$$\frac{P(q_1,q_2,q_3)}{\sqrt{2\pi}}= \exp(q_3^2/2)\int_0^\infty \int_0^{2\pi} \rho^2 
\exp\left(-\rho^2/2 - \rho (q_1\cos(\theta)+q_2\sin(\theta))\right)
\,{\rm d} \theta\,{\rm d}\rho.$$
We note that for any fixed $(q_1,q_2)$ we can find $\Delta\theta$ so that $$q_1\cos(\theta) + q_2 \sin(\theta) = q_{1,2}\cos(\theta'),$$
where $\theta'=\theta+\Delta\theta$  and
$q_{1,2} :=\sqrt{q_1^2+q_2^2}$.
In these new variables,
$$\frac{P(q_1,q_2,q_3)}{\sqrt{2\pi}}= \exp(q_3^2/2)
\int_0^\infty \int_0^{2\pi}\rho^2
\exp\left(-\rho^2/2 - \rho\, q_{1,2} \cos(\theta')\right)
\,{\rm d} \theta'\,{\rm d}\rho.$$
Performing the inner integral, we find
$$\frac{P(q_1,q_2,q_3)}{(2\pi)^{3/2}}= \exp(q_3^2/2)
\int_0^\infty \rho^2e^{-\rho^2/2}\,I_0(\rho\, q_{1,2})\,{\rm d}\rho,$$
where $I_0(\rho\, q_{1,2})$ is the modified Bessel function of the first kind.

Next, we use identity 6.618.4 from Gradshteyn and Ryzhik \cite{GR}:
$$\int_0^\infty e^{-\lambda x^2}I_\nu(\eta x) \,{\rm d}x = \frac{\sqrt{\pi}}{2\sqrt{\lambda}} e^{\eta^2/(8\lambda)} I_{\frac{1}{2} \nu}(\eta^2/(8\lambda)).$$
After differentiating with respect to $\lambda$ and then setting $\lambda = 1/2$ and $\nu =0,$ we obtain
$$\int_0^\infty x^2 e^{- x^2/2}I_0(\eta x) \,{\rm d}x =\frac{{\sqrt{2\pi}}}{4}e^{\eta^2/4}\left[{\eta^2}\left(I_0(\eta^2/4)+I_1(\eta^2/4)\right)+2 I_0(\eta^2/4)\right]. $$

Putting this together with our previous expression for $P,$ we arrive at
\begin{align}\label{a3.1}
\frac{P(q_1,q_2,q_3)}{(2\pi)^2}= \frac 1 4
e^{q_3^2/2}e^{q_{1,2}^2/4}
\left[q_{1,2}^2\left(I_0\left(\frac{q_{1,2}^2}{4} \right)+I_1\left(\frac{q_{1,2}^2}{4} \right) \right)+2I_0 \left(\frac{q_{1,2}^2}{4} \right)  \right].
\end{align}

It follows from \eqref{j13.1} that 
\begin{align*}
-\partial_{q_1} P(q_1,q_2,q_3)
= \int_{\R^3}s\sqrt{b^2+s^2} \exp\left(-\frac{s^2+b^2+w^2}{2}-q_1s-q_2b-q_3w\right)\,{\rm d}s\,{\rm d}b\,{\rm d}w.
\end{align*}
Similar formulas hold if we differentiate with respect to $q_2$ or $q_3$, or we take higher derivatives with respect to these variables. When we evaluate the last expression at $(q_1,q_2,q_3)=(0,0,0)$, we obtain
\begin{align*}
 \int_{\R^3}s\sqrt{b^2+s^2} \exp\left(-\frac{s^2+b^2+w^2}{2}\right)\,{\rm d}s\,{\rm d}b\,{\rm d}w.
\end{align*}
These remarks imply that for any polynomial $p(s,b,w),$ 
\begin{align}\label{j13.3}
    \mathcal{I}_p &:=\int_{\mathbb{R}^3} p(s,b,w)\sqrt{b^2+s^2} 
\exp\left(-s^2/2-b^2/2-w^2/2\right)\,{\rm d}s\,{\rm d}b\,{\rm d}w\\
& = \left. p(-\partial_{q_1},-\partial_{q_2},-\partial_{q_3}) P(q_1,q_2,q_3) \right|_{(q_1,q_2,q_3)=(0,0,0)}.\notag
\end{align}
We now combine this formula and \eqref{a3.1}, and use standard recurrence formulas
for derivatives of the modified Bessel functions to obtain 
\begin{align}
\int_{\mathbb{R}^3} \sqrt{b^2+s^2} \exp\left(-s^2/2-b^2/2-w^2/2\right)\,{\rm d}s\,{\rm d}b\,{\rm d}w &= \frac{(2\pi)^2}{2}, \label{a3.2} \\
\int_{\mathbb{R}^3} w^2\sqrt{b^2+s^2} \exp\left(-s^2/2-b^2/2-w^2/2\right)\,{\rm d}s\,{\rm d}b\,{\rm d}w &= \frac{(2\pi)^2}{2}, \label{a3.3}  \\
 \int_{\mathbb{R}^3} s^2\sqrt{b^2+s^2} \exp\left(-s^2/2-b^2/2-w^2/2\right)\,{\rm d}s\,{\rm d}b\,{\rm d}w &= \frac{3(2\pi)^2}{4},  \label{a3.4} \\
 \int_{\mathbb{R}^3} b^2\sqrt{b^2+s^2} \exp\left(-s^2/2-b^2/2-w^2/2\right)\,{\rm d}s\,{\rm d}b\,{\rm d}w &= \frac{3(2\pi)^2}{4}, \label{a3.5} \\
 \int_{\mathbb{R}^3} s^2 b^2\sqrt{b^2+s^2} \exp\left(-s^2/2-b^2/2-w^2/2\right)\,{\rm d}s\,{\rm d}b\,{\rm d}w&= \frac{15(2\pi)^2}{16},  \label{a3.6} \\
 \int_{\mathbb{R}^3} s^2 w^2\sqrt{b^2+s^2} \exp\left(-s^2/2-b^2/2-w^2/2\right)\,{\rm d}s\,{\rm d}b\,{\rm d}w&=\frac{3(2\pi)^2}{4}.  \label{a3.7}
\end{align}
We now apply \eqref{a3.2}-\eqref{a3.7} to \eqref{j13.2} to obtain
$$I_1(\bal,\bbet) =\sqrt{\frac{3}{2}}\frac{\beta^4}{(2\pi)^\frac{3}{2}\beta_1\beta_2\beta_3} \frac{(2\pi)^2}{2} \left[1+\frac{\eps^2}{2\beta^2}\left(\delta_1^2-\frac{3}{2}\delta_2^2 \right) \right]+O(\eps^3). $$
Expanding the prefactor $\beta^4/(\beta_1 \beta_2\beta_3)$ into a series in $\eps$, we see that
\begin{align}\label{a3.10}
I_1(\bal,\bbet) =\frac{\sqrt{3\pi}\beta}{2} \left[1+\frac{\eps^2}{2\beta^2}\left(\delta_1^2+\frac{\delta_2^2}{2} \right) \right] + O(\eps^3),
\end{align}
which completes the proof of \eqref{j23.1a}.

Note that the integral in \eqref{m12.2a} has the extra factor of $\hat{u}+\hat{y}+\hat{z}$ compared to \eqref{m12.1a}. 
Since $\hat{u}+\hat{y}+\hat{z} = 3\alpha+\sqrt{3}\beta w$, we use \eqref{a1.1} to see that
\begin{align}
I_2(\bal,\bbet) &= \sqrt{\frac{3}{2}} \frac{\beta^4}{(2\pi)^{3/2}\beta_1\beta_2\beta_3}\int_{\mathbb{R}^3} (3\alpha+\sqrt{3}\beta w)\sqrt{b^2+s^2} \notag\\
&\quad\quad\times\left(1+\eps s p_1-(\eps^2/2)(p_2-p_1^2s^2)+\tilde{R}(\eps^3)\right)\, \exp\left(-\omega^2/2\right)\,{\rm d}s\,{\rm d}b\,{\rm d}w,\notag\\
& = 3\alpha I_1 +\frac{3}{\sqrt{2}} \frac{\beta^5}{(2\pi)^{3/2}\beta_1\beta_2\beta_3}\int_{\mathbb{R}^3} w\sqrt{b^2+s^2}\label{a3.11}\\
&\quad\quad\times\left(1+\eps s p_1-(\eps^2/2)(p_2-p_1^2s^2)+\tilde{R}(\eps^3)\right)\, \exp\left(-\omega^2/2\right)\,{\rm d}s\,{\rm d}b\,{\rm d}w.\notag
\end{align}
Let $\wt I_2$ denote the second term in the previous expression. Looking at the parity of each term in the integrand, we see that
\begin{align*}
\wt I_2&(\bal,\bbet) =  \\
&- \sqrt{\frac{3}{2}} \frac{4\eps^2\delta_1\delta_2 \beta^3}{(2\pi)^{3/2}\beta_1\beta_2\beta_3}\int_{\mathbb{R}^3} w^2\sqrt{b^2+s^2}(1-s^2) \exp\left(-\omega^2/2\right)\,{\rm d}s\,{\rm d}b\,{\rm d}w+O(\eps^3 ).
\end{align*}
Using \eqref{a3.3} and \eqref{a3.7}, it follows that
\begin{align*}
\wt I_2&= - \frac{(2\pi)^2}{2}\sqrt{\frac{3}{2}} \frac{4\eps^2\delta_1\delta_2 \beta^3}{(2\pi)^{3/2}\beta_1\beta_2\beta_3}\left(1-\frac{3}{2}\right)+O(\eps^3)
=\eps^2\sqrt{3\pi}\delta_1\delta_2 +O(\eps^3).
\end{align*}
Combining this with \eqref{a3.10} and \eqref{a3.11}, we obtain
\begin{align*}
I_2(\bal,\bbet) = 3\alpha \frac{\sqrt{{3\pi\beta^2}}}{2} \left[1+\frac{\eps^2}{2\beta^2}\left(\delta_1^2+\frac{\delta_2^2}{2} \right) \right] +	\eps^2\sqrt{3\pi}\delta_1\delta_2 +O(\eps^3).
\end{align*}
This completes the proof.
\end{proof}

\subsubsection{Series expansion} This subsection presents a technical result---a series expansion with an explicit error estimate needed in the proof of Lemma \ref{j31.3}.

\begin{lemma}\label{j31.4}
Define $\gamma$ by
\begin{align}\label{eqn:jgh_gam}
\gamma(\eps)=
\gamma(\eps,u,y,z) = \frac{(u+\tilde{\delta}_1\eps -\tilde{\gamma}_1 \eps^2)^2}{(1-\tilde{\delta}_2\eps +\tilde{\gamma}_2 \eps^2)^2}+\frac{(y+2\tilde{\gamma}_1 \eps^2)^2}{(1+2\tilde{\gamma}_2 \eps^2)^2}+\frac{(z-\tilde{\delta}_1\eps -\tilde{\gamma}_1 \eps^2)^2}{(1+\tilde{\delta}_2\eps +\tilde{\gamma}_2 \eps^2)^2}   
\end{align}
and $M$ by
$$M = \max\{1,|\tilde{\delta}_1|,2|\tilde{\gamma}_1|,|\tilde{\delta}_2|, 2|\tilde{\gamma}_2|\}.$$ 
Further suppose that $\eps$ is chosen so that $M \eps \le 1/4.$ Then
\begin{align}\label{a3.30}
    \exp\left(-\gamma(\eps)/2\right) = \left(1+sp_1 \eps-(\eps^2/2)({p}_2-{p_1} s^2)\right) \, \exp\left(-\rho^2/2\right) + \tilde{R},
\end{align}
where
$$|\tilde{R}| \le\frac{(500 \eps M)^3}{6}  (\rho^2+1)^3 e^{- \rho^2/9} ,$$
with $\rho^2 = u^2+y^2+z^2,$ and where
\begin{align*}
p_0 &:= w-\frac{b}{\sqrt{2}},\\
p_1 &:= \tilde{\delta}_1\sqrt{2} +\sqrt{\frac{8}{3}}\tilde{\delta}_2 p_0,
\end{align*}
$$p_2:=\left[2\tilde{\delta}_1^2 +\frac{8}{\sqrt{3}}\tilde{\delta}_1\tilde{\delta}_2 p_0 +\sqrt{24} \tilde{\gamma}_1   b +(3\tilde{\delta}_2^2-2\tilde{\gamma}_2)\rho^2+(2\tilde{\gamma}_2-\tilde{\delta}_2^2)(\sqrt{2}b+w)^2 \right] $$
and $$s=(z-u)/\sqrt{2},\quad\quad b = (y-(u+z)/2) \sqrt{2/3},\quad\quad {\rm and}\quad \quad w=(u+y+z)/\sqrt{3}.$$
\end{lemma}

\begin{proof}
The proof consists of a sequence of estimates for derivatives of elementary functions. 
Their proofs, elementary but tedious, are omitted. 

Let
\begin{align}\label{a3.20}
f(\eps) = (1+a_2\eps +b_2 \eps^2)^{-2}.
\end{align}
If $ \eps \leq 1/(4M_1)$ 
where
$$M_1 = \max\{1,|a_2|,|b_2|\}$$
then
$$ |f(\eps)| \le 4,$$
$$|f'(\eps)|\le  4(12) M_1,  $$
$$|f''(\eps)| \le 4 (12)^2 M_1^2, $$
$$|f'''(\eps)| \le 4 (12)^3 M_1^3.$$

Let
\begin{align}\label{a3.21}
g(\eps) = (a_0  - a_1\eps-b_1\eps^2)^2.
\end{align}
If $ \eps \leq 1/(4M_2)$ 
where
$$M_2 = \max\{1,|a_1|,|b_1|\}$$ 
then
$$|g(\eps)| \le 4(a_0 ^2+1),$$
$$|g'(\eps)|  \le  4 (12) M_2(a_0 ^2+1),$$
$$|g''(\eps)|\le 4 (12)^2 M_2^2(a_0 ^2+1),$$
$$|g'''(\eps)|\le 4(12)^3 M_2^3(a_0 ^2+1).$$

Let $f(\eps)$ and $g(\eps)$ be as in \eqref{a3.20}-\eqref{a3.21}. If $ \eps \leq 1/(4M_3)$ 
where
$$M_3 = \max\{1,|a_1|,|b_1|,|a_2|,|b_2|\}$$ 
then
$$|(fg)'(\eps)| \le 400 (a_0 ^2+1) M_3,$$ 
$$|(fg)''(\eps)| \le 1000\,(a_0 ^2+1)M_3^2,$$ 
$$|(fg)'''(\eps)| \le 250000 \, (a_0 ^2+1)M_3^3.$$

Let $\gamma$ be as in \eqref{eqn:jgh_gam}.
If $ \eps \leq 1/(4M)$ 
where
$$M = \max\{1,|\tilde{\delta}_1|,2|\tilde{\gamma}_1|,|\tilde{\delta}_2|, 2|\tilde{\gamma}_2|\}$$ 
then
$$|\gamma'(\eps)| \le 400 (u^2+y^2+z^2+3) M,$$ 
$$|\gamma''(\eps)| \le 1000\,(u^2+y^2+z^2+3)M^2,$$ 
$$|\gamma'''(\eps)| \le 250000 \, (u^2+y^2+z^2+3)M^3.$$

If $ \eps \leq 1/(4M)$ 
then
$$\left|\frac{{\rm d}^3}{{\rm d}\eps^3} e^{-\gamma(\eps,u,y,z)/2} \right|  
\le \frac{500^3}{2} M^3 (\rho^2+1)^3 e^{-\gamma(\eps,u,y,z)/2}.$$
It can be shown that
$$\gamma \ge \frac{2}{9}\left[{u^2}+y^2+z^2 -3 \right] $$
so we obtain a new estimate for the derivative, 
$$\left|\frac{{\rm d}^3}{{\rm d}\eps^3} e^{-\gamma(\eps,u,y,z)/2} \right| 
\le {500^3} M^3 (\rho^2+1)^3 e^{- (u^2+y^2+z^2)/9}.$$
 The proposition follows  from this estimate, noting that the right-hand side of \eqref{a3.30} is (apart from the remainder) the quadratic Taylor approximation to the left-hand side about $\eps=0$.
\end{proof}

\section{Acknowledgments}

We are grateful to Pablo Ferrari, Tomasz Komorowski, Adam Ostaszewski, John Sylvester and Balint T\'oth for very helpful advice.

We thank an anonymous referee for many suggestions for improvement.



\bibliographystyle{alpha}
\bibliography{pde}

@article{BBO2006,
  title = {Momentum Conserving Model with Anomalous Thermal Conductivity in Low Dimensional Systems},
  author = {Basile, Giada and Bernardin, C\'edric and Olla, Stefano},
  journal = {Phys. Rev. Lett.},
  volume = {96},
  issue = {20},
  pages = {204303},
  numpages = {4},
  year = {2006},
  month = {May},
  publisher = {American Physical Society},
}

@article {BBO2009,
    AUTHOR = {Basile, Giada and Bernardin, C\'edric and Olla, Stefano},
     TITLE = {Thermal conductivity for a momentum conservative model},
   JOURNAL = {Comm. Math. Phys.},
  FJOURNAL = {Communications in Mathematical Physics},
    VOLUME = {287},
      YEAR = {2009},
    NUMBER = {1},
     PAGES = {67--98},
}

@article {JKO,
    AUTHOR = {Jara, Milton and Komorowski, Tomasz and Olla, Stefano},
     TITLE = {Superdiffusion of energy in a chain of harmonic oscillators
              with noise},
   JOURNAL = {Comm. Math. Phys.},
  FJOURNAL = {Communications in Mathematical Physics},
    VOLUME = {339},
      YEAR = {2015},
    NUMBER = {2},
     PAGES = {407--453},
     }

@article {GS,
    AUTHOR = {Grisi, Rafael M. and Sch\"utz, Gunter M.},
     TITLE = {Current symmetries for particle systems with several
              conservation laws},
   JOURNAL = {J. Stat. Phys.},
  FJOURNAL = {Journal of Statistical Physics},
    VOLUME = {145},
      YEAR = {2011},
    NUMBER = {6},
     PAGES = {1499--1512},
}

@article {FT,
    AUTHOR = {Fritz, J\'ozsef and T\'oth, B\'alint},
     TITLE = {Derivation of the {L}eroux system as the hydrodynamic limit of
              a two-component lattice gas},
   JOURNAL = {Comm. Math. Phys.},
  FJOURNAL = {Communications in Mathematical Physics},
    VOLUME = {249},
      YEAR = {2004},
    NUMBER = {1},
     PAGES = {1--27},
}

@misc{KBJS,
      title={Coupled transport equations with freezing}, 
      author={Krzysztof Burdzy and John Sylvester},
      year={2024},
      note={Arxiv 2406.02707},
}

@article {KBAO,
    AUTHOR = {Burdzy, Krzysztof and Ostaszewski, Adam J.},
     TITLE = {Freezing in space-time: a functional equation linked with a
              {PDE} system},
   JOURNAL = {J. Math. Anal. Appl.},
  FJOURNAL = {Journal of Mathematical Analysis and Applications},
    VOLUME = {524},
      YEAR = {2023},
    NUMBER = {2},
     PAGES = {Paper No. 127018, 11},
}

@article {FerOl,
    AUTHOR = {Ferrari, Pablo A. and Olla, Stefano},
     TITLE = {Macroscopic diffusive fluctuations for generalized hard rods
              dynamics},
   JOURNAL = {Ann. Appl. Probab.},
  FJOURNAL = {The Annals of Applied Probability},
    VOLUME = {35},
      YEAR = {2025},
    NUMBER = {2},
     PAGES = {1125--1142},
}

@article {FerEt,
    AUTHOR = {Ferrari, Pablo A. and Franceschini, Chiara and Grevino, Dante G. E. and Spohn, Herbert},
     TITLE = {Hard rod hydrodynamics and the {L}\'evy {C}hentsov field},
   JOURNAL = {Ensaios Matem\'aticos},
  FJOURNAL = {Ensaios Matem\'aticos},
    VOLUME = {38},
      YEAR = {2023},
     PAGES = {185--222},
}

@book {Smoller,
    AUTHOR = {Smoller, Joel},
     TITLE = {Shock waves and reaction-diffusion equations},
    SERIES = {Grundlehren der mathematischen Wissenschaften [Fundamental
              Principles of Mathematical Sciences]},
    VOLUME = {258},
   EDITION = {Second},
 PUBLISHER = {Springer-Verlag, New York},
      YEAR = {1994},
}

@book {Serre,
    AUTHOR = {Serre, Denis},
     TITLE = {Systems of conservation laws. 1},
      NOTE = {Hyperbolicity, entropies, shock waves,
              Translated from the 1996 French original by I. N. Sneddon},
 PUBLISHER = {Cambridge University Press, Cambridge},
      YEAR = {1999},
}

@article {TV03,
    AUTHOR = {T\'{o}th, B\'{a}lint and Valk\'{o}, Benedek},
     TITLE = {Onsager relations and {E}ulerian hydrodynamic limit for
              systems with several conservation laws},
   JOURNAL = {J. Statist. Phys.},
  FJOURNAL = {Journal of Statistical Physics},
    VOLUME = {112},
      YEAR = {2003},
    NUMBER = {3-4},
     PAGES = {497--521},
}

@book {GR,
    AUTHOR = {Gradshteyn, I. S. and Ryzhik, I. M.},
     TITLE = {Table of integrals, series, and products},
   EDITION = {Seventh},
      NOTE = {Translated from the Russian,
              Translation edited and with a preface by Alan Jeffrey and
              Daniel Zwillinger,
              With one CD-ROM (Windows, Macintosh and UNIX)},
 PUBLISHER = {Elsevier/Academic Press, Amsterdam},
      YEAR = {2007},
}

@article{Gaspard_2008,
	doi = {10.1088/1367-2630/10/10/103004},
	url = {https://doi.org/10.1088/1367-2630/10/10/103004},
	year = 2008,
	month = {oct},
	publisher = {{IOP} Publishing},
	volume = {10},
	number = {10},
	pages = {103004},
	author = {Pierre Gaspard and Thomas Gilbert},
	title = {Heat conduction and {F}ourier's law in a class of many particle dispersing billiards},
	journal = {New Journal of Physics},
	}

@article{Gaspard_2,
  title = {Heat Conduction and {F}ourier's Law by Consecutive Local Mixing and Thermalization},
  author = {Gaspard, P. and Gilbert, T.},
  journal = {Phys. Rev. Lett.},
  volume = {101},
  issue = {2},
  pages = {020601},
  numpages = {4},
  year = {2008},
  month = {Jul},
  publisher = {American Physical Society},
  doi = {10.1103/PhysRevLett.101.020601},
  url = {https://link.aps.org/doi/10.1103/PhysRevLett.101.020601}
}

@article{Gaspard_2008_1,
	doi = {10.1088/1742-5468/2008/11/p11021},
	url = {https://doi.org/10.1088/1742-5468/2008/11/p11021},
	year = 2008,
	month = {nov},
	publisher = {{IOP} Publishing},
	volume = {2008},
	number = {11},
	pages = {P11021},
	author = {Pierre Gaspard and Thomas Gilbert},
	title = {On the derivation of {F}ourier's law in stochastic energy exchange systems},
	journal = {Journal of Statistical Mechanics: Theory and Experiment},
}

@article {AHRV,
    AUTHOR = {Angel, Omer and Holroyd, Alexander E. and Romik, Dan and
              Vir\'{a}g, B\'{a}lint},
     TITLE = {Random sorting networks},
   JOURNAL = {Adv. Math.},
  FJOURNAL = {Advances in Mathematics},
    VOLUME = {215},
      YEAR = {2007},
    NUMBER = {2},
     PAGES = {839--868},
}

@article {DF,
    AUTHOR = {Dobrushin, R. L. and Fritz, J.},
     TITLE = {Non-equilibrium dynamics of one-dimensional infinite particle
              systems with a hard-core interaction},
   JOURNAL = {Comm. Math. Phys.},
  FJOURNAL = {Communications in Mathematical Physics},
    VOLUME = {55},
      YEAR = {1977},
    NUMBER = {3},
     PAGES = {275--292},
}

@article {BDS,
    AUTHOR = {Boldrighini, C. and Dobrushin, R. L. and Sukhov, Yu. M.},
     TITLE = {One-dimensional hard rod caricature of hydrodynamics},
   JOURNAL = {J. Statist. Phys.},
  FJOURNAL = {Journal of Statistical Physics},
    VOLUME = {31},
      YEAR = {1983},
    NUMBER = {3},
     PAGES = {577--616},
}

@book {Spohn,
    AUTHOR = {Spohn, Herbert},
     TITLE = {Large scale dynamics of interacting particles},
 PUBLISHER = {Springer-Verlag, Berlin},
      YEAR = {1991},
}

@book {KL,
    AUTHOR = {Kipnis, Claude and Landim, Claudio},
     TITLE = {Scaling limits of interacting particle systems},
    SERIES = {Grundlehren der mathematischen Wissenschaften [Fundamental
              Principles of Mathematical Sciences]},
    VOLUME = {320},
 PUBLISHER = {Springer-Verlag, Berlin},
      YEAR = {1999},
}

@article {Bert,
    AUTHOR = {Bertini, Lorenzo and De Sole, Alberto and Gabrielli, Davide
              and Jona-Lasinio, Giovanni and Landim, Claudio},
     TITLE = {Macroscopic fluctuation theory},
   JOURNAL = {Rev. Modern Phys.},
  FJOURNAL = {Reviews of Modern Physics},
    VOLUME = {87},
      YEAR = {2015},
    NUMBER = {2},
     PAGES = {593--636},
}

@article {AAV,
    AUTHOR = {Amir, Gideon and Angel, Omer and Valk\'{o}, Benedek},
     TITLE = {The {TASEP} speed process},
   JOURNAL = {Ann. Probab.},
  FJOURNAL = {The Annals of Probability},
    VOLUME = {39},
      YEAR = {2011},
    NUMBER = {4},
     PAGES = {1205--1242},
}

@article {AHR,
    AUTHOR = {Angel, Omer and Holroyd, Alexander and Romik, Dan},
     TITLE = {The oriented swap process},
   JOURNAL = {Ann. Probab.},
  FJOURNAL = {The Annals of Probability},
    VOLUME = {37},
      YEAR = {2009},
    NUMBER = {5},
     PAGES = {1970--1998},
}

@article {fold,
    AUTHOR = {Athreya, Jayadev S. and Burdzy, Krzysztof and Duarte,
              Mauricio},
     TITLE = {On pinned billiard balls and foldings},
   JOURNAL = {Indiana Univ. Math. J.},
  FJOURNAL = {Indiana University Mathematics Journal},
    VOLUME = {72},
      YEAR = {2023},
    NUMBER = {3},
     PAGES = {897--925},
}

 \end{document}